\def\submit{0}

\documentclass{article}
\usepackage[utf8]{inputenc}
\usepackage{fullpage,amsmath,amsfonts,amsthm}
\usepackage{xcolor}
\usepackage{hyperref}
\hypersetup{
	colorlinks=true,
	citecolor= blue,
}
\usepackage{thmtools, thm-restate}

\newtheorem*{question}{Question}
\ifnum\submit=0
\newcommand{\dan}[1]{{\color{red}[{\bf DM:} #1]}}
\newcommand{\elchanan}[1]{{\color{blue}[{\bf EM:} #1]}}
\newcommand{\madhu}[1]{{\color{olive}[{\bf MS:} #1]}}
\newcommand{\omri}[1]{{\color{purple}[{\bf OB:} #1]}}
\newcommand{\note}[1]{{\color{orange}[{\bf Note:} #1]}}
\else  
\newcommand{\dan}[1]{}
\newcommand{\elchanan}[1]{}
\newcommand{\madhu}[1]{}
\newcommand{\omri}[1]{}
\newcommand{\note}[1]{}
\fi

\newcommand{\PF}{\textsf{PF}}
\newcommand{\CF}{\textsf{CF}}
\newcommand{\PT}{\textsf{PT}}
\newcommand{\CT}{\textsf{CT}}
\newcommand{\True}{{\textsf{True}}}
\newcommand{\False}{{\textsf{False}}}
\newcommand{\temp}{\textsf{temp}}

\newcommand{\LeavesComponentsPotential}{{\Phi_{\text{LC}}}}

\newcommand{\AdapatedLeavesComponentsPotential}{{\Phi_{\text{LC}}}}
\newcommand{\ExponentialPotential}{{\Phi_{\exp{}}}}
\newcommand{\AdaptedExponentialPotential}{{\widetilde{\Phi}_{\exp{}}}}

\newcommand{\tree}{\mathcal{T}}
\newcommand{\N}{\mathbb{N}}
\newcommand{\E}{\mathbb{E}}
\newcommand{\PP}{\mathbb{P}}

\newcommand{\eps}{\varepsilon}

\newcommand{\C}{{\mathcal{C}}}
\newcommand{\CPT}{{\mathcal{C}_{\text{PT}}}}
\newcommand{\ECR}{{E_{\text{check-root}}}}

\parskip=\medskipamount

\title{Is This Correct? Let's Check!}
\author{%
	Omri Ben-Eliezer\thanks{Department of Mathematics, Massachusetts Institute of Technology, Cambridge, Massachusetts, USA. Email: \texttt{omrib@mit.edu}}
	\and Dan Mikulincer\thanks{Department of Mathematics, Massachusetts Institute of Technology, Cambridge, Massachusetts, USA. Supported in part by a Vannevar Bush Faculty Fellowship ONR-N00014-20-1-2826. Email: \texttt{danmiku@mit.edu}} 
	\and Elchanan Mossel\thanks{Department of Mathematics, Massachusetts Institute of Technology, Cambridge, Massachusetts, USA. Supported in part by a Simons 
		Investigator Award, Vannevar Bush Faculty Fellowship ONR-N00014-20-1-2826, 
		ARO MURI W911NF1910217 and NSF awards 
		DMS-2031883 and CCF 1918421. Email: \texttt{elmos@mit.edu}} 
	\and Madhu Sudan\thanks{School of Engineering and Applied Sciences, Harvard University, Cambridge, Massachusetts, USA. Supported in part by a Simons Investigator Award and NSF Award CCF 2152413. Email: \texttt{madhu@cs.harvard.edu}.}
}

\date{}

\begin{document}
	
	\newtheorem{theorem}{Theorem}[section]
	\newtheorem{definition}[theorem]{Definition}
	\newtheorem{observation}[theorem]{Observation}
	\newtheorem{claim}[theorem]{Claim}
	\newtheorem{lemma}[theorem]{Lemma}
	\newtheorem{proposition}[theorem]{Proposition}
	\newtheorem{corollary}[theorem]{Corollary}

	\maketitle
	
	\begin{abstract}
		Societal accumulation of knowledge is a complex process.  The correctness of new units of knowledge depends not only on the correctness of new reasoning, but also on the correctness of old units that the new one builds on. The errors in such accumulation processes are often remedied  by error correction and detection heuristics. 
		Motivating examples include the scientific process based on scientific publications, and software development based on libraries of code. 
		
		Natural processes that aim to keep errors under control, such as peer review in scientific publications, and testing and debugging in software development, would typically check existing pieces of knowledge -- both for the reasoning that generated them and the previous facts they rely on. In this work, we present a simple process that models such accumulation of knowledge and study the persistence (or lack thereof) of errors.
		We consider a simple probabilistic model for the generation of new units of knowledge based on the preferential attachment growth model, which additionally allows for errors. Furthermore, the process includes checks aimed at catching these errors. We investigate when effects of errors persist forever in the system (with positive probability) and when they get rooted out completely by the checking process. 
		The two basic parameters associated with the checking process are the {\em probability} of conducting a check and the {\em depth} of the check. We show that errors are rooted out if checks are sufficiently frequent and sufficiently deep. In contrast, shallow or infrequent checks are insufficient to root out errors. 
		
	\end{abstract}

	\section{Introduction}
	
	Understanding the robustness of systems to errors is one of the main goals of theoretical computer science. One set of examples lies within information and coding theory, which study this question for electronic information transmission processes. Another important example is quantum computing; the empirical success of this field crucially relies on the difficult challenge of controlling errors in quantum computers. 
	
	In this work, we focus on yet another area where errors are prevalent, and error correction has an important everyday role: \emph{societal knowledge accumulation}. 
	Accumulation of knowledge in the modern world is a very rapid yet noisy process, prone to significant errors as new units of knowledge are established \cite{sep-scientific-knowledge-social}. This, in turn, requires proper error mitigation strategies. For instance, the scientific publication process is based upon the assumption that peer reviewing is able to identify errors in submitted papers, ensuring that (for the most part) the scientific literature remains correct and well-founded. This assumption is however very problematic~\cite{Ioannidis2005,IannidisSelfCorrect2012}: there are numerous examples of important works with a huge impact on the scientific community, whose findings were later found to be completely incorrect, either because of errors or malicious actions, deeming decades of subsequent research essentially useless. One very recent and prominent example is in the study of Alzheimer's disease \cite{blots22, SelkoeCummings22}, where researchers have found evidence that some of the most influential works in the field may be fabricated. 
	
	Knowledge accumulation processes such as the scientific process or software development are based on incremental advances that add to the body of knowledge. 
	Each new unit of knowledge may rely on previously discovered ones. Each proclaimed new unit of knowledge may potentially be erroneous on its own or because it relies on an erroneous unit. 
	Without some checks, errors can overwhelm such cumulative processes. 
	As anyone who has ever developed software knows, debugging and testing are crucial for developing reliable code. Similarly, it is unreasonable to trust scientific discoveries in areas where reviews and replication are not taken seriously (see more below).  
	
	In these and other areas, natural mechanisms
	for checking have been introduced. 
	Our work is motivated by the goal of quantifying the success of such procedures: How should such checking mechanisms be measured and evaluated? Which indicators may suggest a proclaimed fact is likely to be true? What steps would be efficient and effective if one had control of the checking process? 
	Of course, we also do not want to spend too many resources on checking as this slows down the accumulation process, so identifying the sweet spot, where errors are rooted out without spending too many resources on checking, is perhaps the ultimate goal. 
	
	In this work, we study these questions under a very simple model.  This model 
	which we call the \emph{Cumulative Knowledge Process} (CKP) includes certain ingredients addressing the following fundamental questions (which are essential in any knowledge accumulation process):
	\begin{enumerate}
		\item How is knowledge generated and represented?
		\item How do errors arise?
		\item When checks occur, what do they check for and what do they do in case an error is found?
	\end{enumerate}
	
	We describe the precise model that we work with in Section \ref{ssec:model}. This model involves some choices and here we describe the issues and try to explain our choices. 
	For question (1) above, it is natural to represent the body of knowledge as a \emph{directed acyclic graph} (DAG), where units of knowledge are represented as nodes, and an edge from $u$ to $v$ indicates that $v$ ``builds upon'' or ``inherits from'' $u$. Indeed, since $v$ can only build upon units of knowledge that were created before it, the directed graph describing such relations must be a DAG. Now, when a new node $v$ arrives, we need to pick a subset of ``parents'' that $v$ shall build upon.
	The choice of parents should take into account the relevance of the previous node to the new node, which is correlated with latent features such as topics, language, or goals, and also to the importance or impact of the previous node. The former notion (relevance) is somewhat hard to model --- multiple models exist and the best choices are still up for debate. The latter is more familiar with models such as preferential attachment forming a good starting point.  
	
	In this work, we bypass the challenge of assessing relevance by considering a simpler model where knowledge is represented by a {\em tree}, i.e., every new node has only one parent (the challenge of relevance emerges only when we try to model the set of parents, and in particular in determining the correlation between children of two or more parents). In the tree model, we can bypass this issue and simply consider the setting where a newly generated node picks its (unique) parent based on the preferential attachment model for trees. 
	
	Thus in our model, the body of knowledge is a rooted tree, where edges are always directed away from the root and higher degree nodes are more likely to be connected to. We acknowledge that this choice is limiting and more general and realistic DAG processes for cumulative knowledge growth should be studied in future work (we comment more on this, as well as the subsequent work \cite{brandenberger2023combinative}, in Section \ref{sec:open_questions}).

	We now turn to question (2), i.e, the model of errors. Here we introduce the so-called ``primary erroneous facts'' in our model by allowing every newly generated node to be an erroneous one with some fixed probability $\eps$ (this is one of three parameters that specify our process). Erroneous nodes remain hidden until some checking process reveals the error. Till such stage, erroneous processes continue to produce offspring at the same rate as true nodes. Such ``children'' nodes and their descendants in the tree are also erroneous and we refer to them as secondary errors. 
	
	Finally, turning to question (3), i.e., the checking process, we couple the checking of facts with the generation of new children. When a new child is generated, with some probability $p$ a ``check'' is performed. In this check, a path of length $k$ is checked for any primary erroneous node. If any such node is found, all its descendants along the path being checked are ``discovered'' to be erroneous and no longer participate in the growth.

	Our model above captures many natural ingredients of cumulative knowledge (modulo the issue of knowledge being a tree), while still being simple enough to allow for analytic studies. While the error and checking models also involve multiple choices, most are natural (and perhaps not new). The checking process however merits further discussion. First, the model associates checking with the growth of the tree --- checks start at new nodes. This, we feel, reflects a natural choice empirically. Units that are not relevant and not used by others are less frequently checked.
	This is true both in scientific publications and in software development. 
	Indeed, the empirical nature of checking is that facts get checked with probability growing with their impact. In our model, the only impact of a node is in the subtree it generates (and the fraction of the subtree that is not publicly known to be erroneous) and so it makes sense to check only when this impact grows. This leads us to the choice of checking only up to a bounded depth $k$. Checking all ancestors of a node may make checking too expensive. Our choice ensures that every new node seems to add a ``constant'' amount of work independent of the size and shape of the tree, making the checking a plausible process.

	Some of the basic questions that one might want to study about errors in societal knowledge can be posed in this model. In this paper, we introduce some phenomena one might wish to study, such as: ``Does the effect of an error survive for long, or is it fleeting?'', and ``What would be the characteristics of a CKP that would deem it reliable?'', see Definition~\ref{def:CKPprops}. Our results (see exact formulations in Section \ref{sec:results}) can roughly be classified into two types:
	\begin{itemize}
		\item Qualitative results: we identify two contrasting regimes of error propagation. In the first, nodes carrying false information are guaranteed, with probability $1$, to have a finite number of descendants. In the second regime, errors may propagate ad infinitum, and false nodes will exist that serve as roots of trees whose size grows to infinity.
		\item Quantitative results: within the regimes described above we prove quantitative bounds on the types of error. Thus, even when a false tree grows to infinity we show that, depending on the parameters, its size cannot be too large. Moreover, in the setting where false trees are finite almost surely, we demonstrate a very desirable property, most nodes in the tree carry truthful information. 
	\end{itemize}
	
	\subsection{Related Work}
	
	\paragraph{Noisy Computation and the PMC Model.}
	The question of how to address errors in computation has been extensively studied in many communities. In von Neumann's model of noisy computation~\cite{vonNeumann:56}, the model describes noisy gates and it is shown that with sufficient duplication and provided that the error rate is sufficiently small, errors can be controlled by duplicating gates, see also \cite{EvansSchulman:99}. 
	The noisy computation model is in some sense stronger than the one considered here as it considers gates with multiple inputs while in our model each unit depends directly only on one unit. 
	However, the noisy computation model 
	has a central planner that can duplicate many copies of the same computation, while in ours such a planner does not exist. Moreover,  the noisy computation model has an $\omega(1)$ bigger circuit than the noiseless circuit  while the number of operations in our model is only $O(1)$ compared to the model without errors. 
	
	Another extensively studied model of error correction that was introduced is called the PMC model~\cite{PMC:67}. The goal of the PMC model is for functional units to detect the (adversarially) faulty units and this is shown to be achievable under various conditions, see~\cite{AMP:20} and the references within. 
	Again in the PMC model, it is assumed that the structure of the network can be designed beforehand. Moreover, in our model a node may be incorrect due to  errors in previous generations and these could not be checked by immediate neighbors in models of this type. 
	
	\paragraph{Local Error Correction and the ``Positive Rate" Conjecture.}
	
	There are related models of ``memory" on graphs with noisy gates. The goal of these models is to remember a single bit forever using noisy gates. Again in these models, the errors and the checking are very local (depends just on the immediate neighbors), see e.g.,~\cite{gray2001reader,MaMoPo:20}.

	\paragraph{The Reproducibility and Replication Crisis.} 
	
	There is a large body of work indicating that a substantial fraction of published scientific research is incorrect, as it cannot be replicated or reproduced, see, e.g.,~\cite{IannidisSelfCorrect2012,Ioannidis2005,Grcar2013,LarcombeEnvironmental2018}.
	The scientific community is trying to come up with standards and protocols that will reduce that fraction, see e.g.,~\cite{allison2018reproducibility}.
	However, the study of errors in scientific literature mostly ignores the problem of research that is incorrect because it is \emph{based on} incorrect prior research, as these dependencies are hard to understand and control. 
	Our results provide a theoretical framework for addressing this issue.

	\paragraph{Knowledge Aggregation vs. Information Spreading.}
	
	We finally comment on one recent work~\cite{information-spread} by a subset of the current authors that considers a similar process where some information spreads noisily through a network with some local checking. In the setting of \cite{information-spread}, some node in a given network receives a piece of knowledge and spreads it through the network by local communication. Errors in this setting arise from communication errors when transmitting information, and the checking procedure aims to check the local consistency of knowledge. The paper studied the rate at which potentially erroneous information spreads through the network, as opposed to the spread of the corrected information. 
	
	The model in \cite{information-spread} shares similarity with our work in that errors, once generated, may spread and affect other parts of the collective knowledge/belief. However, from this point on, the settings diverge. In particular, a central component of our model is the knowledge graph (the tree) which grows with the process, whereas in \cite{information-spread} the network is extraneous. The checking also models somewhat different settings, and in particular, in our case, when a node is proclaimed to be erroneous, it is actually erroneous (while nodes that proclaim themselves to be true may later end up being found to be false). In contrast, such one-sided guarantees do not hold in the previous work. Finally, the nature of the questions explored in the two models is quite different.

	\subsection{Models}
	\label{ssec:model}
	Here we give the necessary detail and background for our model. The {\em Cumulative Knowledge Process (CKP)} has states $X_0,X_1,\ldots,X_t,\ldots$ where $X_t$ is given by a finite rooted tree $\tree_t$ with each node being given a label from the set $\{\PF,\CF,\CT\}$. We refer to any such labeled tree as a {\em knowledge state}. 
	
	\paragraph{Semantics:} 
	$\CT$ and $\CF$ stand for ``Conditionally True'' and ``Conditionally False'' respectively. A node $v$ represents true knowledge if all nodes on the path from $v$ to the root (including $v$ and the root) are $\CT$. We refer to such a node as a $\True$ node. All other nodes are $\False$ nodes. $\PF$ nodes, for ``proclaimed false'', are those that are publicly false.
	A priori the $\CT$ vs. $\CF$ values are not ``public'' (but can be checked with effort) and so given a node the observable state is either $\PF$ or $\PT$ (for ``proclaimed true'') where $\PT$ is the observation associated with both labels in $\{\CF,\CT\}$. (Formally the observation is given by the function $O:\{\PF,\CT,\CF\} \to \{\PF,\PT\}$ with $O(\PF) = \PF$ and $O(\CT) = O(\CF) = \PT$). 
	
	\paragraph{Parameters:}
	The CKP has three parameters: $\eps \in [0,1]$ denoting the probability of introducing new errors, $p \in [0,1]$ denoting the probability of checking, and $k \in \N$ denoting the length of the check. We note that in principle, $k$ may itself be a random variable. For simplicity, we focus on the case where $k$ is constant. However, as will become evident in our results, not much generality is lost by this choice. Given these parameters, we refer to the process $X_k$ as an $(\eps,p,k)$-CKP. 
	
	\paragraph{State evolution:} Given a state $X_{t}$ at time $t$, the state at time $t+1$, i.e., $X_{t+1}$ is obtained by the following stochastic process: 
	\begin{itemize}
		
		\item\textbf{Choosing a parent.}
		If every node in $\tree_t$ is $\PF$ then the process ``stops'', i.e., $X_{t+1} = X_t$. Else
		first a random $\PT$ node $u$ is selected from the tree $\tree_t$ with probability proportional to $1 + \deg_{\PT}(u)$, where $\deg_{\PT}(u)$ denotes the number of $\PT$ children of $u$. This is reminiscent of the preferential attachment random tree model. The main difference lies in the fact that once a node has been identified as $\PF$, it may not generate new children. As we shall see, this can drastically change the evolution of the CKP, when compared to preferential attachment trees.

		\item\textbf{Creating a new child.}
		A new leaf $v$ is attached as a child to $u$. Set $\tree_{t+1}$ to be $\tree_t$ with the added leaf $v$.
		
		\item \textbf{Error introduction.} $v$ is given the label $\CT$ with probability $1-\eps$
		and the label $\CF$ with probability $\eps$. Let $X_{\temp}$ denote the new knowledge state. 
		
		\item \textbf{Error correction.} With probability $p$, the path with $k$ edges is ``checked'' and if an error (either a $\PF$ or $\CF$ node) is found then all descendants on the path are proclaimed false. Specifically, let $v_0 = v$, $v_1,\ldots,v_{k}$ denote the path of length $k$ starting at $v$. (i.e., $v_i$ is the parent of $v_{i-1}$). If all vertices $v_i$ are labelled $\CT$, then $X_{t+1} = X_{\temp}$. Else let $j$ be the smallest index such that $v_i$ is not labeled $\CT$. We modify $X_{\temp}$ by relabelling $v_{j'} = \PF$ for every $0 \leq j' \leq j$. That is, the entire path between $v$ and $v_j$ is labeled $\PF$. $X_{t+1}$ is the resulting knowledge state. 
	\end{itemize}
	
	At this point we emphasize the fact that errors can propagate in two ways: either a new $\CF$ node is added to the tree in the \emph{Error Introduction} phase, or a $\CT$ node is added as the child of a \False\ parent whose observable state is $\PT$.

	\paragraph{Initial state:} The process we care about starts with $X_0$ being a single root labelled $\CT$ (in other words, the root is always $\True$). This initialization leads to a dichotomy in the tree, \False\ nodes are those nodes which have a $\CF$ ancestor, added at some \emph{error introduction} phase, while $\True$ nodes are connected to the root by a path of $\CT$ nodes. Our aim is to study the difference in behaviors between these two sets.\\
	There is one exception to this initialization rule which we discuss now.
	
	\paragraph{The simple CKP:} To build some intuition and to simplify the proofs we consider a simplified version of the $(\eps,p,k)$-CKP process, in which we set $\eps = 0$. We call such a process a {\em $(p,k)$-simple CKP}. To avoid trivialities, we always set $X_0$ to be a single $\CF$ root in the simple process. Thus, in this process only $\CT$ nodes are added to the tree, all nodes are $\False$, and no new errors are introduced during the process's evolution. This restriction introduces a top-down directionality to the process, since a node may be labeled as $\PF$ only after the same happens to its parent.

	\paragraph{Phenomena we care about:}
	We now make some definitions to describe the different types of behaviors demonstrated by our results.

	\begin{definition} [Properties of CKPs] \label{def:CKPprops}
		For $\eps,p \in [0,1]$ and $k \in \N$, let $\tree_t$ be the knowledge state of an $(\eps,p,k)$-CKP. For a $\CF$ node $u$ added to the tree at some time $t_u$, and for any time $t \geq t_u$, let $\tree_t^u$ denote the sub-tree process rooted at $u$ at time $t$. 
		\begin{itemize}
			\item {\bf Survival of error effects:} 
			We say that a CKP exhibits {\em survival of error effects} if the following holds:
			With positive probability, there exists some $\CF$ node $u$ for which the process $\tree_t^u$ goes forever. That is, for every $t$, there is at least one $\PT$ node in $\tree_t^u$. 
			\item {\bf Elimination of error effects:} We say that a CKP exhibits {\em elimination of error effects} if it does not exhibit survival of error effects. Specifically for every $\CF$ node $u$, there exists a time $t$ after which every node in $\tree_t^u$ is $\PF$.  In particular, in the simple process, elimination of the error effects means that all nodes have become $\PF$ and so the process stops.
			\item {\bf $\varphi$-reliable process:} Let $\varphi:\N \to \N$ satisfy $\varphi(t)=o(t)$. We say the CKP process is \emph{$\varphi$-reliable} when the following holds. Let $M_t$ stand for the maximal size of a sub-tree in $\tree_t$ of $\False$ and $\PT$ nodes. Then, for any constant $c>0$,
			$$\PP\left(M_t > \varphi(t)\right) \xrightarrow{t\to \infty}0.$$
			Since at time $t$ the tree always has $t$ nodes this condition means that any \False{} and $\PT$ node, can only have a negligible number of descendants. We will usually take $\varphi(t) = \Theta(t^a)$, for some $a <1$.
			\item\textbf{$\delta$-highly reliable process:} Let $\delta > 0$. We say that the CKP process \emph{$\delta$-highly reliable} if, for any time $t$, the expected proportion of $\False$ nodes labeled as $\PT$, as opposed to $\True$ nodes, is at most $\delta.$ We say that a process is highly reliable if for every $\delta > 0$ it is $\delta$-highly reliable. (In other words, a process is highly reliable if most of the nodes that declare themselves as $\True$ are indeed $\True$.)
		\end{itemize}
		Observe that all definitions above are only a function of the parameters $\eps$, $p$, and $k$.
	\end{definition}
	We note that, other than the fact that survival and elimination of error effects are mutually exclusive, it is not a-priori clear that one definition implies or denies another. Our main results will identify regimes of parameters where the CKP (or its simple variant) satisfies one, or more, of these definitions.
	
	There are additional natural questions one may ask about the model. For example, one could consider highly noisy regimes, when the proportion of \False\ nodes with a $\PT$ label is $1-o(1)$, and so most of the information carried by the process is corrupted. We leave such questions for future investigations.
	
	\paragraph{\PT\ components:}
	We finish this section with the following basic definition, of a proclaimed true (\PT) component. This refers to a sub-tree $T$ in any of our processes (CKP or simple CKP), that at some time $t$ consists only of proclaimed true nodes; and furthermore, is maximal with respect to this property. 
	
	\begin{definition}[\PT\ Component]
		Consider any of the cumulative knowledge processes we define (CKP or simple CKP) at some time $t \geq 0$, and let $\tree_t$ denote the (undirected version of the) tree generated by the process. We call the sub-graph of $\tree_t$ consisting of all $\False$ nodes whose observable state is Proclaimed True (\PT), simply, the \PT\ sub-graph. A connected component of the \PT\ sub-graph is called a \PT\ component. We denote by $\CPT(t)$ the set of all \PT\ components in the process at time $t$.
		
		If $C\in \CPT(t)$ is a $\PT$ component, we denote by $|C|$ the number of $\PT$ nodes in $C$, and if $C' \in \CPT(t+1)$ we use the notation $C' \subset C$ to indicate that $C'$ was created from $C$ in one step. Formally, this means that the root of $C'$ is a descendent of the root of $C$.
	\end{definition}
	
	As we shall see soon, \PT\ components play an essential role in many of our arguments. In particular, we base several potential functions used in our proofs on these components.

	\subsection{Proof Ideas and Techniques}
	Our model is based on the preferential attachment model, which has been studied for nearly a century, originating in the work of Yule, \cite{yule1925mathematical}. By now, the model is well-understood and the growth and evolution of many quantities of interest have been thoroughly analyzed. Thus, our choice of the model should allow us to tap into many known results, and moreover, the recursive nature of the tree is expected to often lend itself to an exact analysis of the dynamics. 
	
	Having mentioned the above, we now remark that the addition of the checking procedure to our model can be thought of as a destructive process, competing with the natural growth process of the preferential attachment tree. While the recursive nature of our process still exists, the dynamics of the preferential attachment tree are often distorted, and some of the underlying symmetry is broken. This becomes particularly evident when one looks at the graph structure of $\PT$ nodes, which will now form a forest of fractured components, rather than an ever-expanding tree. Still, we show that the model lends itself to the probabilistic analysis of several relevant functionals, fundamental to our analysis. Most of the functionals are defined as a sum of simpler functionals applied to individual components. We consider functionals that take into account the number of leaves, the degrees of nodes and their depth, the size of components, etc.
	
	The functionals are analyzed by probabilistic techniques in particular using the analysis of (sub/super)-martingales. Roughly speaking, once we show that the processes have a drift in a certain direction, these techniques allow making asymptotic conclusions about the process. Of course, coming up with the right functionals requires creativity and intuition about various aspects of our process.

    \paragraph{Acknowledgements:} We thank Anna Brandenberger and Peter Gacs for spotting some mistakes in an earlier version. We are also grateful to Luke Podwojski and Serte Donderwinkel for a careful reading of the paper and for identifying and correcting several errors in the proofs, which are addressed in this revision.
    
	\section{Our Results} \label{sec:results}
	We now describe our results. We begin by addressing the simple CKP model and then proceed by establishing analogs for the general model.
	
	\subsection{Results for the Simple Model}
	Our first two results show that, depending on the parameters $p$ and $k$, error effects can both survive and be eliminated, in the simple model.
	
	\begin{restatable}[Error effect elimination in the simple model]{theorem}{simpleerrorelim}
		\label{thm:error_effect_elim_simple}
		For all $p \geq \frac{6}{7}$ and $k \geq 5$, the error effect in the $(p,k)$-simple CKP is completely eliminated.
	\end{restatable}
	\begin{restatable}[Error effect survival in the simple model]{theorem}{simpleerrorsurvive}
		\label{thm:simpleerrorsurvival}
		For all $0 < p \leq \frac{1}{4}$ and $1< k \leq \infty$, the error effect in the $(p,k)$-simple CKP survives with positive probability.
	\end{restatable}
	The two theorems above demonstrate contrasting behaviors, which mainly depend on $p$, the probability of checking for errors. Thus, if one wants to ensure complete elimination of errors, one should be willing to look for errors in a reasonable proportion of knowledge units. 
	
	Let us note that the two regimes in Theorems \ref{thm:error_effect_elim_simple} and \ref{thm:simpleerrorsurvival} do not cover the entire parameter space. The behavior of the process for intermediate $p \in (\frac{1}{4}, \frac{6}{7})$ is an interesting question which is left open. In particular, identifying the critical $p$ in which there is a phase transition between survival and elimination of the error effect would be appealing.
	
	To address the role of the parameter $k$ in this model, we focus on the assumption $k\geq 5$ in Theorem \ref{thm:error_effect_elim_simple}. Remarkably this is not a technical issue, and the model displays a striking transition when $k$ is small. We show that small changes to the depth of error correction can have dramatic effects.
	
	\begin{restatable}[Error effect survival when $k = 2$]{theorem}{notwolevels}
		\label{thm:notwolevels}
		For any $0\leq p < 1$ the error effects in the $(p,2)$-simple CKP survive with positive probability.
	\end{restatable}
	
	Our next result further elucidates the role of $k$ in the model by examining the reliability of the process.  Specifically, we show that even when the error effect in the simple model survives with positive probability, for example, when $p \leq \frac{1}{4}$, the process can still be \emph{reliable}, provided $k$ is large enough.
	\begin{theorem} \label{thm:reliableprocess}
		Let $p \in (0,1)$. If $k \in \N$ is such that $\frac{12}{2k-3} \leq p$ then the $(p,k)$-simple CKP is $\varphi$-reliable, with $\varphi(t)=\Theta(t^{0.55})$.
	\end{theorem}
	Thus, if $k$ is large, even if a few units of knowledge are periodically checked, it can still be ensured that no false source will corrupt a non-negligible portion of the knowledge base.
	
	At this point, we do not know whether an absolute constant (such as $0.55$) is the correct term in the exponent in Theorem \ref{thm:reliableprocess}.
	In fact, it is reasonable to expect that the actual power will depend on $k$ and perhaps also on $p$.
	However, below, in Theorem \ref{thm:bigcomponents}, we shall show that one cannot expect better than polynomial dependency and that with constant probability, there will exist ($\False$) $\PT$ components of polynomial size.
	
	\subsection{Results for the General Model}
	In the general model, we first generalize Theorems \ref{thm:simpleerrorsurvival} and \ref{thm:error_effect_elim_simple}. Because of the added parameter $\eps$, the overall dependence on the other parameters becomes slightly more complicated. For now, it will suffice to state simplified versions of the results. The reader is referred to Section \ref{sec:generalproofs} for the exact dependencies.
	\begin{restatable}[Error effect elimination in the general model]{theorem}{generalerrorelim}
		\label{thm:error_effect_elim_general}
		For every $\eps \in (0,1)$, there exists $p_0 \in (0,1)$ and $k_0 \in \N$, such that for any $p \geq p_0$ and $k \geq k_0$, the error effects in the $(\eps,p,k)$-CKP are completely eliminated.
	\end{restatable}
	\begin{restatable}[Error effect survival in the general model]{theorem}{generalerrorsurvive}
		For every $\eps \in (0,1)$, there exists $p_0 \in (0,1)$, such that for any $k \in \N$ and $p \leq p_0$, the error effects in the $(\eps,p,k)$-CKP survives.
		\label{thm:generalerrorsurvival}
	\end{restatable}
	Like in the simple model, we see the critical role $p$ plays. At least for low levels of error, one can always invest appropriate effort by performing enough checks and guaranteeing low-term correctness of the information state in the tree. On the other hand, if $p$ is small enough, the process could get overwhelmed by errors.
	
	Using the same ideas used to generalize Theorem \ref{thm:simpleerrorsurvival} into Theorem \ref{thm:generalerrorsurvival}, an analog of Theorem \ref{thm:reliableprocess} could also be derived for the general model. However, due to the dependencies introduced by new $\CF$ nodes, the obtained result is somewhat complicated and not immediately interpretable. Instead, we prove another desirable property that emerges when the error effect is eliminated. Namely, we show that, not only do false sub-trees get eliminated, but that the overall proportion of $\True$ nodes in the process remains large. Thus, we show that the process is highly reliable.
	
	\begin{theorem}\label{thm:highreliableprocess}
		Under the same conditions of Theorem \ref{thm:error_effect_elim_general}, if $p \geq p_0$, and $k \geq k_0$, then the $(\eps,p,k)$-CKP is $O(\eps(1-p))$-highly reliable.
	\end{theorem}
	Note that by a given time $t$, we should expect to add $\eps(1-p)\cdot t$ $\CF$ nodes to the tree; the $(1-p)$ factor comes from the fact that a check is performed a new node is automatically labeled as $\PF$ and can thus be disregarded. Thus, $\eps(1-p)$ is the proportion of errors introduced to the tree without accounting for further propagation via attaching new descendants. With this in mind, one way to interpret Theorem \ref{thm:highreliableprocess} is as a statement about types of errors in the process; most errors are expected to come from very shallow sub-trees, and errors, when introduced, tend to be quickly rectified before spreading along the process.
	
	Let us note that there is no hope in improving Theorem \ref{thm:highreliableprocess} by more than a constant. To see this, for a given time $t > 0$, it's enough to consider all the $\CF$ nodes added after time $\frac{t}{2}$. It can be shown that with some constant probability, each such node will not spawn a descendent, and hence will survive up to time $t$. Thus, the expected proportion of $\CF$ nodes will be $\Omega(\eps(1-p))$. 
	\section{Proofs for the Simple Model}
	
	\subsection{Error Effect Elimination in the Simple Model}
	Our first aim is to prove that when $p$ is large, the error effects in the simple model are completely eliminated. For convenience, we restate the theorem.
	\simpleerrorelim*
	
	Towards the proof of Theorem \ref{thm:error_effect_elim_simple} we define our first potential function, the exponential potential.
	
	\begin{definition}[Exponential potential]\label{def:exponential_potential}
		Consider any of the CKP models at time $t \geq 0$. The \emph{root} $r$ of a PT component $C \in \CPT(t)$ is the oldest (earliest birth) node in $C$. The \emph{depth} $|v|$ of a node $v$ in $C$ is defined as the length (number of edges) of the shortest path between $v$ and $r$ in $C$. The \emph{degree} of $v$ in $C$ is defined as $d(v) = 1 + \deg_{\PT}(v)$. 
		
		Finally, we define the \emph{exponential potential} of $C$ as above by 
		$$
		\ExponentialPotential(C) = \sum_{v \in C} d(v) \cdot 2^{|v|},
		$$
		and the potential over the whole tree $\tree_t$ at time $t$ as the sum over all $C \in \CPT(t)$,
		$$
		\ExponentialPotential(\tree_t) = \sum_{C \in \CPT(t)}\ExponentialPotential(C) .
		$$
	\end{definition}
	
	The main component of the proof is the following lemma, asserting that the exponential potential decreases in expectation, i.e., the sequence of potentials at any time along the process is a super-martingale. 
	\begin{lemma}
		\label{lem:exponential_potential_contracts}
		For all $p \geq \frac{6}{7}$, $k \geq 5$, and tree $\tree_t$ representing the knowledge state of the $(p,k)$-simple CKP, the following holds. If $\ExponentialPotential(\tree_t)>0$, then
		$$
		\E\left[\ExponentialPotential(\tree_{t+1}) \ | \ \tree_t\right] < \ExponentialPotential(\tree_t).
		$$
	\end{lemma}
	We first complete the proof of Theorem \ref{thm:error_effect_elim_simple} given the lemma. 
	\begin{proof}[Proof of Theorem \ref{thm:error_effect_elim_simple}]
		Let $\{\tree_t\}_{t=1}^{\infty}$ be the sequence of knowledge state trees in the CKP. From Lemma \ref{lem:exponential_potential_contracts} we deduce that the sequence $\{\ExponentialPotential(\tree_t)\}_{t=1}^{\infty}$ is a positive super-martingale. Thus, by the martingale convergence theorem, \cite[Theorem 4.2.11]{durrett2016probability}, there exists a random variable $\ExponentialPotential(\tree_\infty)$, such that $\ExponentialPotential(\tree_t) \xrightarrow{t\to \infty}\ExponentialPotential(\tree_\infty)$ almost surely. We claim that $\PP\left(\ExponentialPotential(\tree_\infty) = 0\right) = 1$. Indeed, if $\ExponentialPotential(\tree_t) \neq 0$, then 
		$$|\ExponentialPotential(\tree_t) - \ExponentialPotential(\tree_{t+1})| \geq 1,$$
		which implies that $0$ is the only possible limit.
		Also observe that $|\tree_t|\leq \ExponentialPotential(\tree_t)$. Combining everything, we see, that almost surely,
		$$|\tree_t|\leq \ExponentialPotential(\tree_t) \xrightarrow{t\to \infty}0,$$
		which is the claim.
	\end{proof}
	
	It remains to prove Lemma \ref{lem:exponential_potential_contracts}, about the contraction (in expectation) of the exponential potential.
	\begin{proof}[Proof of Lemma \ref{lem:exponential_potential_contracts}]
		
		Let $C \in \CPT(t)$ be the $\PT$ component to which the new node in the process connects, and let $\sum\limits_{\substack{C'\in \CPT(t+1)\\C' \subset C}}C'$ denote the result of this connection. 
		Note that all components of $\tree_t$ other than $C$ remain unchanged in the transition to $\tree_{t+1}$, and so it suffices to analyze $\E\left[\sum\limits_{\substack{C'\in \CPT(t+1)\\C' \subset C}}\ExponentialPotential(C') - \ExponentialPotential(C)\right]$.
		Finally, define $D = \sum_{v \in C} d(v)$ and for each $v \in C$ set $q(v) = d(v) / D$. 
		
		\paragraph{Case I: $|C| \geq k$.} Suppose first that $C$ contains at least $k$ nodes. Since the root $r$ of a $\PT$ component is itself proclaimed true, the nearest node that a check starting below $r$ can detect as false is the parent of $r$, which lies one level above $r$. Hence a check reveals the root only if it reaches depth one greater than the distance from the newly added node to $r$. In this case, the probability that a node of distance less than $k-1$ from the root $r$ of $C$ is selected is at least $(2k-3)/D$, since there are at least $k-1$ such nodes, and at least $k-2$ edges in the connected component containing $r$ and all of these low-depth nodes. 
		
		Denote by $\ECR$ the event that the newly added node $u$ is of distance at most $k-1$ from $r$, and furthermore, $u$ is checked. Observe that $\PP(\ECR) \geq p \cdot (2k-3)/D$. When this event occurs, the check reaches the parent of $r$, which is proclaimed false, and subsequently, all nodes on the path between (and including) $r$ and $u$ are marked $\PF$. The depth of all other nodes decreases by at least one. Thus, on the event $\ECR$ we always have $\sum\limits_{\substack{C'\in \CPT(t+1)\\C' \subset C}}\ExponentialPotential(C') < \frac12 \cdot \ExponentialPotential(C)$, and therefore
		\begin{align*}
		\E\left[\Bigg(\sum\limits_{\substack{C'\in \CPT(t+1)\\C' \subset C}}\ExponentialPotential(C') - \ExponentialPotential(C)\Bigg){\bf 1}_{\ECR} \ \Big| \ \tree_t \right] &\leq -\frac12 \ExponentialPotential(C)\, \PP(\ECR) \\ &\leq -\frac{p}{2} \cdot \frac{2k-3}{D}\, \ExponentialPotential(C).
		\end{align*}
		On the complementary event, a new node $u$ is added with parent $v$. The total contribution to the potential is at most $2^{|v|+1} + 2^{|v|}$, resulting from the depth of $u$ and the added degree of $v$. Since the parent is $v$ with probability $q(v)$ and this contribution is non-negative, bounding ${\bf 1}_{\neg\ECR} \leq 1$ gives
		\begin{align*}
		\E\left[\Bigg(\sum\limits_{\substack{C'\in \CPT(t+1)\\C' \subset C}}\ExponentialPotential(C') - \ExponentialPotential(C)\Bigg){\bf 1}_{\neg\ECR} \ | \ \tree_t \right] &\leq \sum_{v \in C} q(v) \cdot (2^{|v|+1} + 2^{|v|}) \\
		&= 3 \cdot \sum_{v \in C} \frac{d(v)}{D} \cdot 2^{|v|} = \frac{3}{D} \cdot \ExponentialPotential(C).
		\end{align*}
		Combining the above two contributions,
		\begin{align}\label{eq:exppotsuperfirst}
		\E\left[ \sum\limits_{\substack{C'\in \CPT(t+1)\\C' \subset C}}\ExponentialPotential(C') - \ExponentialPotential(C) \ | \ \tree_t \right]   
		&< \frac{\ExponentialPotential(C)}{D} \cdot \left( -\frac{1}{2} \cdot (2k-3)p + 3 \right).
		\end{align}
		The second term is non-positive if $(2k-3)p \geq 6$. Therefore, the conditions $p \geq \frac{6}{7}$ and $k \geq 5$ guarantee the non-positivity of the latter.
		
		\paragraph{Case II: $|C| < k$.} In the case where $C$ contains less than $k$ nodes, we know that all nodes have depth less than $k$, and so checking (if happens) will always reach the root and mark it $\PF$.
		Thus, similarly to above, with probability $p$ the depth of all surviving nodes in $C$ decreases by at least one, so
		$\ExponentialPotential(C') < \ExponentialPotential(C) / 2$. 
		If checking does not happen, the potential increases in expectation by an additive factor of at most $\frac{3}{D}\cdot \ExponentialPotential(C)$ similarly to Case I. Thus,
		\begin{equation} \label{eq:exppotsmallcomps}
		\E\left[ \sum\limits_{\substack{C'\in \CPT(t+1)\\C' \subset C}}\ExponentialPotential(C') - \ExponentialPotential(C) \ | \ \tree_t \right] < 
		\ExponentialPotential(C) \left[-\frac{p}{2} + \frac{3}{D}(1-p)\right] \leq
		\frac{\ExponentialPotential(C)}{D}\left[-\frac{p}{2} + 3(1-p)\right],
		\end{equation}
		where the RHS is non-positive for $p \geq \frac{6}{7}$.
		
	\end{proof}

	\subsection{Survival of Error Effects in the Simple Model}
	Next, we define another potential, which adds up the number of leaves of \PT\ components and the number of components.
	
	\begin{definition}[Leaves and components potential]\label{def:comps_leaves_potential}
		Consider any of the CKP models at time $t \geq 0$. A node $v$ in a component $C \in \CPT(t)$ is considered a leaf if it does not have any descendent in $C$ (we note that a leaf is allowed to have descendants not in $C$; in particular it might have proclaimed false children). 
		
		For a given component $C \in \CPT(t)$, the \emph{leaves and components potential} restricted to $C$ is $1$ if $|C|=1$, and otherwise it is one plus the number of leaves in $C$. Finally, the leaves and components potential $\LeavesComponentsPotential(\tree_t)$ is the sum of potentials of all $C \in \CPT(t)$.
	\end{definition}
	
	The leaves and components potential is used to show that for a small enough checking probability $p$, the error effect can survive forever with positive probability. This is the content of Theorem \ref{thm:simpleerrorsurvival}, which we now restate. 
	\simpleerrorsurvive*
	The main ingredient of the proof is showing that the potential grows to infinity in expectation and has bounded differences. This is formalized in the following claim.
	\begin{lemma} \label{lem:leavespotentialevol}
		Consider the $(p,k)$-Simple CKP process, let $t \in \N$, and suppose that $\LeavesComponentsPotential(\tree_t) > 0$. Denote by $\Delta_t = \LeavesComponentsPotential(\tree_{t+1}) - \LeavesComponentsPotential(\tree_t)$ the change in the leaves and components potential at time $t$. Then $\Delta_t \geq -2$ always holds and when $\LeavesComponentsPotential(\tree_t) > 0$ we have
		$$\E[\Delta_t | \tree_t] > \frac{1}{2}-2p.$$
	\end{lemma}
	Note that the process stops if the potential reaches zero: if $\LeavesComponentsPotential(t) = 0$ then $\LeavesComponentsPotential(t') = 0$ for any $t' > t$.
	\begin{proof}[Proof of Lemma \ref{lem:leavespotentialevol}]
		Let $v$ denote the node added to the process at time $t+1$. Recall that its parent $u$ must be a proclaimed true node, and let $C \in \CPT(t)$ be the PT component containing $u$ at time $t$. 
		Note that attaching $v$ to $C$ does not modify any of the PT components $C' \neq C$. This follows from the next observation.
		\begin{observation}
			Let $v$ be a proclaimed false node in the simple CKP. Then all ancestors of $v$ in the process are also proclaimed false.
		\end{observation}
		Thus, the potential of any $C' \in \CPT(t)$ other than $C$ remains unchanged at time $t+1$; it, therefore, suffices to analyze the change of potential in $C$.
		
		First, if $|C| = 1$, then $|C \cup \{v\}| = 2$. Suppose first that $v$ does not run the checking procedure (or runs it but does not reach a PF node); this holds with probability at least $1-p$. Since $C \cup \{v\}$ contains precisely one leaf in this case, its total potential is $2$, an increase of $1$ over the potential of $C$.
		In the other case, with probability at most $p$, checking takes place and both $u$ and $v$ are marked PF, thus removing $C$ from $\CPT$ without creating new PT nodes, which decreases the total potential by $1$. In total, the expected change in potential is at least $1 \cdot (1-p) - 1 \cdot p = 1-2p$.
		
		Otherwise, $|C| > 1$, and $\deg_{\PT}(\mathrm{root}) \geq 1$. Let $\ell$ stand for the number of leaves in $C$ and set $u:= \mathrm{parent}(v)$. Since $v$ chooses the parent $u \in C$ according to a preferential attachment distribution, we have,
		\begin{align*}
		\PP\left(u \text{ is a leaf}\right) &= \frac{\sum\limits_{w\in C \text{ is a leaf}}(\deg_{\PT}(w)+1)}{\sum\limits_{w\in C}(\deg_{\PT}(w)+1)} \\ &= \frac{\sum\limits_{w\in C \text{ is a leaf}}(\deg_{\PT}(w)+1)}{\sum\limits_{w\in C \text{ is a leaf}}(\deg_{\PT}(w)+1) + \sum\limits_{w\in C \text{ is not a leaf}}(\deg_{\PT}(w)+1)} \\ &= \frac{\ell}{\ell + \sum\limits_{w\in C \text{ is not a leaf}}\deg_{\PT}(w) + \sum\limits_{w\in C \text{ is not a leaf}}1} \leq \frac{\ell}{2\ell +1} < \frac{1}{2}.
		\end{align*}
		Here we used that a leaf $w$ has $\deg_{\PT}(w) = 0$ (so the leaf-sum equals $\ell$), that the non-leaf nodes have total $\PT$-degree $|C| - 1 \geq \ell$, and that $C$ has at least one non-leaf.
		Equivalently,
		$\PP\left(u \text{ is not a leaf}\right) > \frac{1}{2}.$
		Note that equality to $\ell / (2 \ell + 1)$ is attained if and only if $C$ is a rooted star. There are several cases to consider.
		\begin{enumerate}
			\item If $u$ is a leaf, and $v$ does not run a check, then the potential remains unchanged, i.e., $\Delta_t = 0$. 
			\item If $u$ is not a leaf, and again, $v$ does not run a check, then $v$ is a new leaf, and the potential increases by one: $\Delta_t = 1$. Thus, by the above,
			$$\PP\left(\Delta_t = 1\right)\geq \PP\left(u \text{ is not a leaf and no check was performed} \right) > \frac{(1-p)}{2}.$$
			\item If $v$ runs a check (with probability $p$), then the potential can decrease in two ways.
			The added node $v$ can remove a leaf from $C$, which can only happen if $u$ is a leaf.
			Note that any removed parent of $u$, other than the root, can only increase the potential, since it would create new connected components, without affecting the number of leaves.
			So, the other possibility to decrease the potential is to remove the root. This can happen regardless of whether $v$ is connected to a leaf or an internal node. 
			Thus,
			\begin{equation} \label{eq:rootremoval}
			    \PP\left(\Delta_t = -2\right)\leq \PP\left(u \text{ is a leaf and a check was performed} \right) =  \PP\left(u \text{ is a leaf} \right)\cdot p,
			\end{equation}
			and 
			$$\PP\left(\Delta_t = -1\right)\leq \PP\left(u \text{ is not a leaf and a check was performed} \right) =  (1-\PP\left(u \text{ is a leaf} \right))\cdot p.$$
		\end{enumerate}
		Denote $\alpha := \PP\left(u \text{ is a leaf} \right)$ and recall, as shown above, that $\alpha < 1/2$. Summarizing the above cases:
		$$ \E\left[\Delta_t|\tree_t\right] \geq \frac{1-p}{2} -2\alpha p -(1-\alpha)p= \frac{1}{2} -(\frac{3}{2}+\alpha)p >  \frac{1}{2} - 2p$$
		where the last expression is non-negative when $p \leq \frac{1}{4}$.
	\end{proof}
	Theorem \ref{thm:simpleerrorsurvival} will now follow by utilizing the following simple fact about sub-martingales.
	\begin{lemma}\label{lem:positivesubmartingale}
		Let $\{X_t\}_{t\geq 0}$ be a non-negative sub-martingale with filtration $\{\mathcal{F}_t\}_{t\geq 0}$ and such that $X_0 > 0$. Assume that there exist constants $c_1,c_2 > 0$, such that, for every $t\geq 0$, when $X_t \neq 0$:
		\begin{enumerate}
			\item $|X_{t+1} - X_t| \leq c_1$ almost surely.
			\item $\E\left[X_{t+1} - X_t|\mathcal{F}_t\right] > c_2$.
		\end{enumerate}
		Then with positive probability, $X_t > 0$ holds for all $t \in \N$.
	\end{lemma}
	\begin{proof}
		Assume for now that $X_0$ is large enough (as a function of $c_1$ and $c_2$). Later we show that this assumption is not needed.
		
		Define the process $Y_t = X_t - tc_2$. It is straightforward to verify that $Y_{0} = X_0$ and $\E[Y_{t+1} | Y_t] \geq Y_t$, so $\{Y_t\}_{t=t_0}^{\infty}$ is a sub-martingale. We apply Azuma's inequality, \cite[Theorem 1.10.30]{doerr},
		to get that
		\begin{align}
		\label{eqn:bound_comps_and leaves}
		\PP\left(X_{t} \leq 0 \right) = \PP\left(Y_{t} \leq tc_2 \right) \leq \exp \left( \frac{-(X_0 +c_2t)^2 }{2c_1^2t}\right).
		\end{align}
		It is not hard to see that if $X_0$ is large enough (as a function of $c_1$ and $c_2$), we have
		$$
		\sum_{t=1}^{\infty} \exp \left( \frac{-(X_0 +c_2t)^2 }{2c_1^2t}\right) < \frac12.
		$$
		Thus, by a union bound over all $t \in \N$,
		$$\PP\left(\min\limits_{t \geq 0} X_t > 0 \right)\geq \frac{1}{2}.$$
		
		We conclude by relaxing the assumption that $X_0$ is large enough. Note that for any constant $C = C(c_1, c_2)$ there exists $t_0 = t_0(c_1, c_2)$ for which with positive probability $X_{t_0} \geq C$. Conditioning on this event and applying the above arguments on the sub-martingale $\{Z_t\}_{t \geq 0}$ defined by $Z_t = X_{t+t_0}$, the proof follows.
	\end{proof}
	\begin{proof}[Proof of Theorem \ref{thm:simpleerrorsurvival}]
		We use Lemma \ref{lem:leavespotentialevol} along with concentration inequalities for martingales with bounded differences.
		Clearly, for any fixed large constant $C_p > 0$ (possibly dependent on $p$), there exists some time $t_0$, depending on $C_p$, for which with probability bounded away from zero, $\LeavesComponentsPotential(\tree_{t_0}) \geq C_p$. Condition on this event (call it $E$), and consider the process $X_t = \LeavesComponentsPotential(\tree_{t_0+t})$. 
        
        In order to apply Lemma \ref{lem:positivesubmartingale}, we need $X_t$ to be a sub-martingale with bounded increments. However, note that while the increments of $X_t$ are bounded from below, by $-2$, it can have arbitrarily large positive increments. These positive increments may occur when a removed root creates many new fragmented components. To handle this technicality, we modify $X_t$ is in the following way. Define a new process $\tilde{X}_t$, such that $\tilde{X}_0 = X_0$ and for $t > 0$,
        $$\tilde{X}_t = \begin{cases}\tilde{X}_{t-1} + X_t - X_{t-1}& \text{if } X_t - X_{t-1} \leq 2\\ 
        \tilde{X}_{t-1} + 2 & \text{if } X_t - X_{t-1} > 2\end{cases}.$$
        Clearly $\tilde{X}_t$ is adapted to same filtration as $X_t$, $\tilde{X}_t \leq X_t$, and $|\tilde{X}_t - \tilde{X}_{t-1}| \leq 2$ almost surely, for every $t > 0$.
        Moreover, since $X_t - X_{t-1} > 2$ only when a root is removed from a $\CT$ component, the proof of Lemma \ref{lem:leavespotentialevol}, or more spefically the argument preceding \eqref{eq:rootremoval}, shows that 
        $$\E\left[\tilde{X}_{t} - \tilde{X}_{t-1}|X_t\right] \geq \frac{1}{2} - 2p.$$
  
  Thus, according to Lemma \ref{lem:leavespotentialevol}, when $p \leq \frac{1}{5}$, $\tilde{X}_t$ satisfies the conditions of Lemma \ref{lem:positivesubmartingale} with $c_1 = 2$ and $c_2 = \frac{1}{2} - \frac{5}{2}p.$
		Since $\PP(E) > 0$ and $\tilde{X}_t \leq X_t$, we have, 
		$$\PP\left(\min\limits_{t\geq 0}\tilde{X}_t > 0\right) > 0 \implies \PP\left(\min\limits_{t\geq 0} X_t > 0\right) > 0\implies \PP\left(\min\limits_{t\geq 0} \LeavesComponentsPotential(\tree_t) > 0\right)>0,$$
		where, by Lemma \ref{lem:positivesubmartingale}, the left expression holds as long as $C_p$ is large enough.
		
		Since $|\tree_t| \geq \LeavesComponentsPotential(\tree_t)$ the proof is complete.
	\end{proof}
	\subsection{Error effect survival for shallow checks}
	In this section, we focus on the case where $k=2$, and hence a new node only checks two levels up. We show in this case that the error effects can survive with positive probability, irregardless of the value of $p$. Let us recall the exact statement.
	\notwolevels*
	Our proof goes by comparing the CKP to an ever-growing branching process. Specifically, we shall say that $v$ is a \emph{univalent} node if $\deg_{\PT}(v) = 1$ and we say that two univalent nodes, $v$ and $u$ are independent if neither is the ancestor of the other. In the proof, we will show that if a CKP is rooted at a univalent node, then by the time the root is labeled as $\PF$, the expected number of independent univalent nodes is larger than $1$. A standard argument then shows that the tree has a positive probability to survive forever. 
	
	To control the expected number of independent univalent nodes, we will need the following combinatorial lemma.
	\begin{lemma}\label{lem:induni}
		Let $T$ be a tree of height $h$ with $n$ univalent nodes. Then, $T$ at least $\frac{n}{h}$ mutually independent univalent nodes.
	\end{lemma}
	\begin{proof}
		We call a sequence of univalent nodes $P = (v_0,v_1,\dots,v_m)$ a univalent path, if it is a sub-sequence of a simple path, starting from $v_0$ and terminating at a leaf. In other words, for $i = 0,\dots, m-1$, $v_i$ is an ancestor of $v_{i+1}$, and no other node from $P$ lies in the path between them. 
		We say that a univalent path is maximal if it is not a strict sub-sequence of another univalent path. For $P$ as above, we call $v_m$ a terminal point.
		
		The main observation is that if $P \neq P'$ are two different maximal univalent paths, with respective terminal points $v$ and $v'$, then necessarily $v$ and $v'$ are independent. Indeed, $P \neq P' \implies v\neq v'$ and by maximality neither $v$ nor $v'$ have any univalent decedents. 
		
		Thus, we define the following set,
		$$A:=\{v\in T| v \text{ is a terminal point of a maximal univalent path}\}.$$ 
		From the above observation, we get that the number of mutually independent univalent nodes in $T$ is at least $|A|$. On the other hand, each univalent path contains at most $h$ nodes, and a simple counting argument gives,
		$$|A| \geq \frac{n}{h}.$$
	\end{proof}
	Next, we analyze the growth of a $\CT$ component in the simple CKP, when it is rooted at a univalent node. For this we make the following definition: a univalent initialization of a CKP is a knowledge state $\tree_0$ such that all nodes in $\tree_0$, except its root $r'$ are labeled $\PT$, with the root being $\PF$. We also require that both $r'$ and its single child $r$ are univalent.
	\begin{lemma} \label{lem:univalentgrowth}
		Let $p \in [0,1]$ and let $\tree_t$ be the knowledge state of a $(p,2)$-simple CKP.
		Assume that $\tree_0$ is a univalent initialization, and define the stopping time,
		$$\tau = \max\{t: \deg_{\PT}(r) = 1\},$$
		where $r$ is the single child of the root.
		Then, for every $t \geq 0$,
		$$\PP\left(\tau = t\right) = \Omega\left(\frac{1}{t^2}\right).$$
		In particular,
		$$\E\left[|\tree_\tau|\right] = \infty.$$ 
	\end{lemma}
	\begin{proof}
		Observe that at step $t$, $|\tree_t| = t + a$, for $a=|\tree_0|$. Recall that a new node chooses its parent among the $\PT$ nodes only, so the normalizing constant of the preferential attachment rule at step $t$ is
		$$S_{t-1} := \sum_{\substack{v \in \tree_{t-1}\\ v \text{ is } \PT}}(\deg_{\PT}(v)+1) = \sum_{v \in \tree_{t-1}}(\deg_{\PT}(v)+1) - \big(\deg_{\PT}(r')+1\big) = 2t+2a-5,$$
		where we used that every node other than the root $r'$ is $\PT$ and has exactly one parent edge, so that, by the handshaking lemma, $\sum_{v \in \tree_{t-1}}(\deg_{\PT}(v)+1) = 2|\tree_{t-1}| - 1 = 2t+2a-3$, and that the root $r'$ is $\PF$ (hence not selectable) with $\deg_{\PT}(r') = 1$. Thus, conditioned on $\tau > t-1$, since $\deg_{\PT}(r) = 1$, we have, from the preferential attachment rule,
		$$\PP\left(\tau = t| \tau > t-1\right) =  \frac{\deg_{\PT}(r) + 1}{S_{t-1}} = \frac{2}{2t + 2a - 5} \geq \frac{1}{t + a}.$$
		Moreover, the same argument also shows,
		$$\PP\left(\tau > t| \tau > t-1\right) =  1 - \frac{2}{2t+2a-5} = \frac{2t+2a-7}{2t+2a-5} \geq \frac{2t-2}{2t} = \frac{t-1}{t},$$
		where the inequality holds because $a \geq \frac{5}{2}$ (indeed, $a \geq 3$ by the univalent initialization).
		By iterating this argument, we get,
		$$\PP\left(\tau > t\right) = \prod\limits_{i=1}^t \PP\left(\tau > i| \tau > i-1\right) \geq \prod\limits_{i=1}^t\frac{i-1}{i} = \frac{1}{t},$$
		where we have used that the product is telescopic. 
		Thus, 
		$$\PP\left(\tau = t\right) = \PP\left(\tau = t| \tau > t-1\right)\cdot \PP\left(\tau > t-1\right) \geq \frac{1}{t+a}\frac{1}{t-1} = \Omega \left(\frac{1}{t^2}\right).$$
		To get the bound on $|T_\tau|$, observe that as long as $\deg_{\PT}(r) = 1$, no new node was attached below $r$ and hence, since $k = 2$, no check could reach the $\PF$ root $r'$. Hence, as the size of the tree grows at each time step, $|\tree_\tau| \geq \tau$. Finally, a straightforward calculation for the expectations shows
		$$\E\left[|\tree_\tau|\right] \geq \E\left[\tau\right] = \sum\limits_{t=1}^\infty t\PP\left(\tau = t\right) \geq \sum\limits_{t=1}^\infty \Omega\left(\frac{1}{t}\right) = \infty.$$
	\end{proof}
	We now prove that when $k=2$, the error effects always survive with positive probability.
	\begin{proof}[Proof of Theorem \ref{thm:notwolevels}]
		Let $r'$ be the $\CF$ root of the simple CKP. Since $p < 1$ there is a positive probability that in $2$ steps, $r'$ will have a single univalent child. We treat this configuration as a univalent initialization, as in Lemma \ref{lem:univalentgrowth} (the fact that $r'$ is $\CF$, instead of $\PF$ does not matter). Let, $r$ be the univalent child of $r'$ and define
		\begin{align*}
		\tau' &= \max\{t: r \text{ is } \PT\},\\
		\tau &= \max\{t: \deg_{\PT}(r) = 1\}.
		\end{align*}
		Since $k=2$, the only way $r$ can become $\PF$ is if a node is attached as its child and performs a check, which realizes $r'$ is false. Thus, $\tau' \geq \tau$.
		Let us denote $h(\tree_{\tau'})$ as the height of $\tree_{\tau'}$ and $\mathrm{uni}(\tree_{\tau'})$ as the number of univalent nodes in $\tree_{\tau'}$.
		Observe that for $t \leq \tau'$, $\tree_t$ evolves like a preferential attachment tree. Thus, we invoke the following structural results, from \cite[Theorem 1.1]{brightwell2012vertices} and \cite[Theorem 1.2]{pain2022correction}, concerning the above parameters:
		\begin{enumerate}
			\item There exists a constant $c_1 > 0$, such that,
			$$\frac{h(\tree_{\tau'})}{\log(\tau')} \xrightarrow{\tau' \to \infty} c_1.$$
			\item There exists a constant $c_2 > 0$, such that,
			$$\frac{\mathrm{uni}(\tree_{\tau'})}{\tau'} \xrightarrow{\tau' \to \infty} c_2.$$
		\end{enumerate}
		Now, let $\mathrm{Iuni}(\tree_{\tau'})$ stand for the maximal number of mutually independent univalent nodes in $\tree_{\tau'}$. From Lemma \ref{lem:induni} we have,
		$$\mathrm{Iuni}(\tree_{\tau'}) \geq \frac{\mathrm{uni}(\tree_{\tau'})}{h(\tree_{\tau'})}.$$
		
		As alluded to previously, we would like to show that the number of mutually independent univalent nodes in each $\PT$ component behaves like a branching process with many offspring. To use such an argument we would need to show that the expectation $\E\left[\mathrm{Iuni}(\tree_{\tau'})\right]$ is large. In light of the above, bounding the expectation will be facilitated by showing that the sequence of random variables $(Z_t)_{t\geq0}:= \left(\frac{\mathrm{uni}(\tree_t)}{h(\tree_t)}\frac{\log(t)}{t}\right)_{t\geq 0}$ is uniformly integrable. 
		By \cite[Lemma 2.5]{dommers2010diameters}, we have for some $\eta,c_1' >0$, that,
		$$\PP\left(h(\tree_t) < c_1'\log(t)\right) \leq \frac{1}{t^\eta}.$$
		Thus, for $M > 0$ large enough, since $\frac{\mathrm{uni}(\tree_t)}{t}\leq 1$,
		$$\E\left[Z_t{\bf1}_{\{Z_t > M\}}\right]\leq \E\left[\frac{\log(t)}{h(\tree_t)}{\bf1}_{\{h(\tree_t) < \frac{\log(t)}{M}\}}\right] \leq \frac{\log(t)}{t^{\eta}} = o(1),$$
		and the sequence $Z_t$ is uniformly integrable.
		
		To bound the expectation of $\mathrm{Iuni}(\tree_{\tau'})$, we now observe that conditional on $\tau > t$, the sub-tree without the root $\tree_t \setminus \{r\}$ evolves like the usual preferential attachment tree.
		Combining this observation with the above estimates, and invoking Vitali's convergence theorem (as in, e.g., \cite{folland1999real}), shows that there exists some $t_0 >0$ and some constant $c_3>0$, such that for any $t \geq t_0$, 
		$$\E\left[\mathrm{Iuni}(\tree_{t}){\bf 1}_{\{\tau' = t\}}\right] \geq c_3\frac{t}{\log(t)}\PP\left(\tau' = t\right).$$
		We thus have,
		\begin{align*}
		\E\left[\mathrm{Iuni}(\tree_{\tau'})\right] & \geq\sum\limits_{t \geq t_0}\E\left[\mathrm{Iuni}(\tree_{t}){\bf 1}_{\{\tau' = t\}}\right] \geq c_3\sum\limits_{t \geq t_0}\frac{t}{\log(t)}\PP\left(\tau' = t\right)\\
		&\geq c_3\sum\limits_{t \geq t_0}\frac{t}{\log(t)}\Omega\left(\frac{1}{t^2}\right) = c_3\sum\limits_{t \geq t_0}\Omega\left(\frac{1}{t\log(t)}\right) =\infty,
		\end{align*}
		where we have used Lemma \ref{lem:induni} for the third inequality.
		
		To finish the proof, at time $\tau'$ we regard each univalent node accounted for in $\mathrm{Iuni}(\tree_{\tau'})$ as an eventual root for a new $\PT$ component. Repeating the above argument for each such component separately yields a branching process on univalent nodes with infinite offspring expectation. In particular, it is standard that this branching process has a positive probability to continue forever, e.g. \cite[Theorem 4.3.12.]{durrett2016probability}. Clearly, if this branching process survives then the same must be true for $\tree_t$ and we conclude the proof.
	\end{proof}
	
	\subsection{No Linear Components}
	We now show that even when the error effect survives,  components in $\CPT(t)$ tend to be sub-linear in $t$, and so the process is \emph{$\varphi$-reliable} for some $\varphi(t) =o(t)$. Formally, we prove the following theorem, which implies Theorem \ref{thm:reliableprocess}.
	
	\begin{theorem} \label{thm:smallcomps}
		For every $k$ large enough, and any $\frac{12}{2k-3} \leq p \leq \frac{1}{6}$,  the error effect in the $(p,k)$-Simple CKP survives with positive probability. Moreover,
		$$\PP\left(\exists C \in \CPT(t) : |C| \geq \sqrt{10}(k+1)\sqrt{2}^kt^{0.55}\right) \leq \frac{1}{t^{0.1}}.$$
	\end{theorem}
	We begin by defining a slight adaptation of the exponential potential.
	\begin{definition}[Adapted exponential potential]\label{def:biased_exponential_potential}
		Consider any of the CKP models at time $t \geq 0$. The \emph{root} $r$ of a PT component $C \in \CPT(t)$ is the oldest (earliest birth) node in $C$. The \emph{depth} $|v|$ of a node $v \in C$ is defined as the length (number of edges) of the shortest path between $v$ and $r$. We also define the \emph{degree} as $d(v):= \deg_{PT}(v) +1$. 
		
		Finally, we define adapted the \emph{exponential potential} of $C$ as above by 
		$$
		\AdaptedExponentialPotential(C) = |C|\sum_{v \in C} d(v) \cdot 2^{|v|}.
		$$
	\end{definition}
	The following lemma shows that the exponential potential decreases in expectation when the component is large. It is an adaption of Case I from Lemma \ref{lem:exponential_potential_contracts} with slightly worse constants.
	\begin{lemma}
		\label{lem:exponential_potential_contracts_small}
		For all $k \geq 2$, $\frac{12}{2k-3} \leq p$, and $C \in \CPT(t)$ a \PT\ components with $|C| > k$
		$$
		\sum\limits_{\substack{C'\in \CPT(t+1)\\C' \subset C}}\E\left[\AdaptedExponentialPotential(C') \ | \ \tree_t\right] < \AdaptedExponentialPotential(C),
		$$
	\end{lemma}
	\begin{proof}
		
		First, define $D = \sum_{v \in C} d(v)$ and for each $v \in C$ set $q(v) = d(v) / D$. As in Case I of Lemma \ref{lem:exponential_potential_contracts}, since the root $r$ of a $\PT$ component is itself proclaimed true, the nearest node a check can detect as false is the parent of $r$, one level above $r$. Observe that, since $|C| > k$, the probability that a node of distance less than $k-1$ from the root $r$ of $C$ is selected is at least $(2k-3)/D$, since there are at least $k-1$ such nodes, and at least $k-2$ edges in the connected component containing $r$ and all of these low-depth nodes. 
		
		Denote by $\ECR$ the event that the newly added node $u$ is of distance at most $k-1$ from $r$, and furthermore, $u$ is checked. By the above $\PP(\ECR) \geq p \cdot (2k-3)/D$. When this event occurs, the check reaches the parent of $r$, which is proclaimed false, and subsequently, all nodes on the path between (and including) $r$ and $u$ are marked $\PF$. The depth of all other nodes decreases by at least one. Moreover, it is clear that if $C' \in \CPT(t+1)$ is a new component $C'\subset C$ in this case, then $|C'| < |C|$. Thus, on the event $\ECR$ we always have $\sum\limits_{\substack{C'\in \CPT(t+1)\\C' \subset C}} \AdaptedExponentialPotential(C') < \frac12 \cdot \AdaptedExponentialPotential(C)$, and therefore
		\begin{align*}
		\E\left[\Bigg(\sum\limits_{\substack{C'\in \CPT(t+1)\\C' \subset C}}\AdaptedExponentialPotential(C') - \AdaptedExponentialPotential(C)\Bigg){\bf 1}_{\ECR} \ \Big| \ \tree_t \right]
		&\leq -\frac12 \AdaptedExponentialPotential(C)\, \PP(\ECR) \\
		&\leq -\frac{p}{2} \cdot \frac{2k-3}{D}\, \AdaptedExponentialPotential(C).
		\end{align*}
		On the complementary event, a new node $u$ is added with parent $v \in C$, and the component containing $u$ at time $t+1$ has at most $|C|+1$ nodes. As in Lemma \ref{lem:exponential_potential_contracts}, the new node and the added degree of $v$ contribute at most $2^{|v|+1} + 2^{|v|}$ to $\sum_{w} d(w)\, 2^{|w|}$, while the size factor in the definition of $\AdaptedExponentialPotential$ increases from $|C|$ to at most $|C|+1$. Hence, on the event $\neg\ECR$,
		\begin{align*}
		&\sum\limits_{\substack{C'\in \CPT(t+1)\\C' \subset C}} \AdaptedExponentialPotential(C') - \AdaptedExponentialPotential(C) \\
		&\qquad\leq (|C|+1)\Bigg(\sum_{w \in C} d(w)\, 2^{|w|} + 2^{|v|+1} + 2^{|v|}\Bigg) - |C| \sum_{w \in C} d(w)\, 2^{|w|} \\
		&\qquad= \frac{\AdaptedExponentialPotential(C)}{|C|} + (|C|+1)\left(2^{|v|+1} + 2^{|v|}\right).
		\end{align*}
		Since the parent is $v$ with probability $q(v)$ and the right-hand side is non-negative, bounding ${\bf 1}_{\neg\ECR} \leq 1$ gives
		\begin{align*}
		&\E\left[\Bigg(\sum\limits_{\substack{C'\in \CPT(t+1)\\C' \subset C}}\AdaptedExponentialPotential(C') - \AdaptedExponentialPotential(C)\Bigg){\bf 1}_{\neg\ECR} \ \Big| \ \tree_t \right] \\
		&\qquad\leq \frac{\AdaptedExponentialPotential(C)}{|C|} + (|C|+1) \sum_{v \in C} q(v)\left(2^{|v|+1} + 2^{|v|}\right) \\
		&\qquad= \frac{\AdaptedExponentialPotential(C)}{|C|} + \frac{3(|C|+1)}{|C|} \cdot \frac{\AdaptedExponentialPotential(C)}{D} \leq \left(\frac{1}{|C|} + \frac{4}{D}\right)\AdaptedExponentialPotential(C),
		\end{align*}
		where we used that $\sum_{v \in C} q(v)\, 2^{|v|} = \frac{1}{D}\sum_{v \in C} d(v)\, 2^{|v|} = \frac{\AdaptedExponentialPotential(C)}{|C|\, D}$, and that $(|C|+1)/|C| \leq 4/3$ since $|C| > k \geq 2$.
		
		Combining the above two contributions, and noting that $D < 2|C|$,
		\begin{align}\label{eq:exppotsuper}
		\E&\left[\sum\limits_{C' \in \CPT(t+1): C'\subset C} \AdaptedExponentialPotential(C') - \AdaptedExponentialPotential(C) \ | \ \tree_t \right] \nonumber  \\
		&\leq \frac{\AdaptedExponentialPotential(C)}{D} \cdot \left( -\frac{1}{2} \cdot (2k-3)p + 4 + \frac{D}{|C|}\right) \nonumber\\
		&< \frac{\AdaptedExponentialPotential(C)}{D} \cdot \left( -\frac{1}{2} \cdot (2k-3)p + 6\right).
		\end{align}
		Thus, the condition $(2k-3)p \geq 12$ ensures that the right-hand side of \eqref{eq:exppotsuper} is non-positive, i.e., that the adapted exponential potential decreases in expectation.
	\end{proof}
	
	We may now prove Theorem \ref{thm:smallcomps}.
	\begin{proof}[Proof of Theorem \ref{thm:smallcomps}]
		Let us define the potential 
		$$\Phi(\tree_t):= \LeavesComponentsPotential(\tree_t) - \frac{1}{5(k+1)^22^k}\sum\limits_{C \in \CPT(t): |C| >k}\AdaptedExponentialPotential(C).$$
		It will suffice to show that $\{\Phi(\tree_t)\}_{t\geq 0}$ is a sub-martingale.
		Indeed, since $\Phi(\tree_0) = 1$, we shall deduce that $\E\left[\Phi(\tree_t)\right] \geq 1$.
		
		By combining the expectation bound with the facts $\LeavesComponentsPotential(\tree_t)  \leq 2t$ and $|C|^2 \leq \AdaptedExponentialPotential(C)$, we obtain,
		$$\E\left[\sum\limits_{C \in \CPT(t): |C| >k}|C|^2\right] \leq \E\left[\sum\limits_{C \in \CPT(t): |C| >k}\AdaptedExponentialPotential(C)\right] \leq \frac{10(k+1)^22^k}{p}t.$$
		
		In particular, by Markov's inequality, this shows for any $C \in \CPT(t)$, and $t > k$, that
		$$\PP\left(\exists C \in \CPT(t) : |C|^2 \geq 10(k+1)^22^kt^{1.1}\right) \leq \PP\left(\sum\limits_{C \in \CPT(t): |C| >k}|C|^2 \geq 10(k+1)^22^kt^{1.1}\right) \leq \frac{1}{t^{0.1}}.$$
		Put differently, we obtain the desired result,
		$$\PP\left(\exists C \in \CPT(t) : |C| \geq \sqrt{10}(k+1)\sqrt{2}^kt^{0.55}\right) \leq \frac{1}{t^{0.1}}.$$
		
		Thus, to finish the proof we show that $\{\Phi(\tree_t)\}_{t\geq 0}$ is a sub-martingale.
		
		We first note that by Lemma \ref{lem:leavespotentialevol}, when $p \leq \frac{1}{4}$ the potential $\LeavesComponentsPotential(\tree_t)$ is a sub-martingale:
		\begin{equation} \label{eq:supersizepot}
		\E\left[\LeavesComponentsPotential(\tree_{t+1}) - \LeavesComponentsPotential(\tree_{t})\ |\ \tree_t\right]>  \frac{1}{2}-2p \geq 0.
		\end{equation}
		Moreover, suppose that $C \in \CPT(t)$, with $|C| = k$, and that $C' \in \CPT(t+1)$ is such that $C' \subset C$ and $|C'| = k+1$. In other words, in step $t+1$ a new node was added to $C$ resulting in a new 
		`large' \PT\ component $C'$. Since $|C'| = k+1$, we have $\AdaptedExponentialPotential(C') = (k+1)\sum\limits_{v \in C'} d(v)2^{|v|} \leq (k+1)^22^k.$
		In this case, if a new node was attached to $C$, the above reasoning shows that
		\begin{align} \label{eq:supermovepot}
		\E\left[\Phi(\tree_{t+1}) -\Phi(\tree_{t}) |\ \tree_t\right] &\geq \E\left[\LeavesComponentsPotential(C') - \LeavesComponentsPotential(C) - \frac{\AdaptedExponentialPotential(C')}{5(k+1)^22^k}\ |\ \tree_t\right] \nonumber\\ 
		&> \frac{1}{2}- 2p - (1-p)\frac{(k+1)^22^k}{5 (k+1)^22^k} = \frac{3-18p}{10} \geq 0,
		\end{align}
		
		where the last inequality holds when $ p \leq \frac{1}{6}$. Finally,
		if $|C| > k$, because $\frac{12}{2k-3} \leq p$, we have by the proof of Lemma \ref{lem:exponential_potential_contracts} that
		\begin{equation} \label{eq:superexppot}
		\E\left[\AdaptedExponentialPotential(C) - \sum\limits_{C' \in \CPT(t+1): C'\subset C}\AdaptedExponentialPotential(C') \ | \ \tree_t\right] > 0.
		\end{equation}
		Summing up \eqref{eq:supersizepot}, \eqref{eq:supermovepot}, and \eqref{eq:superexppot} shows
		$$\E\left[\Phi(\tree_{t+1}) - \Phi(\tree_t)|\tree_t\right] > 0,$$
		which finishes the proof.
	\end{proof}

	\subsection{Lower Bound on Components Size}
	In Theorem \ref{thm:smallcomps} we showed that the size of the largest $\False$ component is sublinear in $t$; the upper bound obtained was polynomial in $t$. Our next result shows that this dependence is essentially correct: the largest component is indeed polynomially sized. 
	\begin{theorem}\label{thm:bigcomponents}
		Suppose that $k\geq 2$ and $p \leq \frac{1}{4}$. Conditional on the tree surviving, there exists a constant $0 < c_{p,k} < 1$, depending only on $p$ and $k$, such that, for every $t > 0$,
		$$\PP\left(\exists t' \leq t\ \exists C \in \CPT(t') : |C| \geq \frac{1}{2}\sqrt{t}^{\frac{0.35}{k}}\right) \geq c_{p,k}.$$
	\end{theorem}
	We begin with a simple lemma, which bounds the probability of a single component being large.
	\begin{lemma} \label{lem:singlecomponentsize}
		For any $k \geq 2$ and $p \in (0,1)$, $$\PP\left(\exists C \in \CPT(t) : |C| = t\right) \geq \frac{c'_{p,k}}{t^k},$$
		for some constant $c'_{p,k}$, depending only on $p$ and $k$.
	\end{lemma}
	\begin{proof}
		We first look at the process at time $k$, and denote by $c'_{p,k}$ the (positive) probability that no check was made and that the tree looks like a path of length $k$. The rest of the proof continues conditional on this event.
		
		Denote that set of nodes in the path, except the node farthest away from the root, by $P$. Note that the size of the tree can only shrink if a new node is added with a parent in $P$.
		For $t > k$, denote by $v_t$, the newly inserted node at time $k$, and by $P(v_t)$ its parent.
		We have,
		$$\PP\left(P(v_{k+1}) \notin P\right) \geq \frac{1}{2k},$$
		and 
		$$\PP\left(P(v_{k+2}) \notin P| P(v_{k+1}) \notin P\right) \geq \frac{2}{2(k+1)}.$$
		By inducting this argument we see, for fixed $t> k$,
		$$\PP\left(P(v_t),\dots,P(v_{k+2}),P(v_{k+1}) \notin P\right) \geq \frac{1}{2}\prod_{i=1}^{t-k}\frac{i}{(k+i)} =\frac{1}{2}\frac{(t-k)!k!}{t!} = \frac{1}{2 \binom{t}{k}} \geq \frac{1}{t^k}.$$
		The proof is complete since in the above event the tree contains a single component of size $t$.
	\end{proof}
	The next lemma shows that by a given time we have many small components, each one can serve as an origin of a new tree to which we may apply Lemma \ref{lem:singlecomponentsize}.
	\begin{lemma}\label{lem:manycomponents}
		For any $k \geq 2$ and $p \leq 1/4$ there exists a constant $c''_{p,k} > 0$, depending only on $p$ and $k$, such that for any $t > c''_{p,k}$, conditional on the tree surviving,
		$$\PP\left(|\CPT(t)| > c''_{p,k} \cdot t^{0.35}\right) \geq c''_{p,k},$$
		where $|\CPT(t)|$ stands for the number of $\PT$ components at time $t$.
	\end{lemma}
	\begin{proof}
		Consider the leaves and components potential $\LeavesComponentsPotential(\tree_t)$. As seen in Lemma \ref{lem:leavespotentialevol}, $\LeavesComponentsPotential(\tree_t)$ is a sub-martingale and, moreover, conditional on the tree surviving,
		$\E\left[\LeavesComponentsPotential(\tree_t)\right] > (\frac{1}{2} - 2p)t.$ 
		Since $\LeavesComponentsPotential(\tree_t) \leq 2t$ always holds, by the reverse Markov inequality,
		$$\PP\left(\LeavesComponentsPotential(\tree_t) \geq (\frac{1}{2} - 2p)t^{0.9}\right) \geq \frac{\E\left[\LeavesComponentsPotential(\tree_t)\right] - (\frac{1}{2} - 2p)t^{0.9}}{2t - (\frac{1}{2} - 2p)t^{0.9}} \geq   (\frac{1}{2} - 2p)\frac{t-t^{0.9}}{2t + 4pt^{0.9}} > \frac{1-4p}{8},$$
		where the last equality holds when $t$ is large enough. Now, by Theorem \ref{thm:smallcomps}, we know that,
		$$\PP\left(\exists C \in \CPT(t) : |C| \geq \sqrt{10}(k+1)\sqrt{2}^kt^{0.55}\right) \leq \frac{1}{t^{0.1}}.$$
		Thus, as long as $t$ is large enough, with probability $\frac{1-4p}{8}$, $\LeavesComponentsPotential(\tree_t) > \frac{1}{2}(1-4p)t^{0.9}$, and there are no components of size larger than $\sqrt{10}(k+1)\sqrt{2}^kt^{0.55}$.
		Since $|C| \geq \LeavesComponentsPotential(C)$, it follows, by a counting argument, that $|\CPT(t)| \geq \frac{1-4p}{8\sqrt{10}(k+1)\sqrt{2}^k}t^{0.35}$.
	\end{proof}
	Having established the existence of many components the final ingredient of the proof is showing that a single component will not be starved out.
	\begin{lemma} \label{lem:leafevolution}
		Let $\tree_{t_0}$ be a $(p,k)$-simple CKP knowledge state tree at time $t_0$ and let $v$ be a leaf in $\tree_{t_0}$. For any $t > t_0$, suppose that $v$ was not removed from the tree and let $\tree^v_{t}$ be the sub-tree rooted at $v$ at time $t$. Then,
		$$\PP\left(|\tree^v_{t_0^2}| \geq \frac{t_0}{2} \right) \geq \frac{1}{8}.$$
		
	\end{lemma}
	\begin{proof}
		Consider a Poly\'a urn such that at time $0$, there are $t_0-1$ black balls and $1$ white ball, and let $X_t$ stand for the number of white balls at time $t.$
		Observe that, if no nodes were removed from $\tree_t$, then $|\tree^v_{t}|$ has the same law as $X_{t-t_0}$. Thus, conditioned on $v$ not being removed, since other nodes may be removed, we have that $|\tree^v_{t}|$ stochastically dominates $X_{t-t_0}$. It is well known that $X_t$ follows a Beta-Binomial distribution with parameters $t$, $t_0-1$, and $1$.
		
		In particular, for $t = t_0^2 - t_0$ we have
		$$\E[X_t] = \frac{t}{t_0} = t_0 - 1,$$
		and
		$$\E[X_t^2] = \frac{t(2t+t_0-1)}{t_0(t_0+1)} = \frac{(t_0^2 - t_0) (2t_0^2 - t_0 - 1)}{t_0(t_0+1)} \leq 2(t_0-1)^2.$$ 
		Thus, conditioned on $v$ not being removed, by the Paley-Zygmund inequality, again for $t = t_0^2-t_0$
		\begin{align*}
		\PP\left(|\tree^v_{t_0^2}| > \frac{t_0-1}{2}\right)
		&\geq
		\PP\left(X_{t} > \frac{t_0-1}{2}\right) = \PP\left(X_{t} > \frac{1}{2}\E[X_{t}]\right)\\
		&\geq \frac{1}{4}\frac{\E\left[X_{t}\right]^2}{\E\left[X^2_{t}\right]} \geq \frac{1}{8}.
		\end{align*}
		Since the size of $\tree^{v}_{t_0^2}$ must be an integer, we conclude that $\PP\left(|\tree^v_{t_0^2}| \geq t_0/2\right) > 1/8$.
	\end{proof}
	
	We are now ready to prove our lower bound.
	\begin{proof}[Proof of Theorem \ref{thm:bigcomponents}]
		Fix $t$, and denote by $E$ the event that $|\CPT(\sqrt{t})| > c_{p,k}''\sqrt{t}^{0.35}$, where $c_{p,k}''$ is as in Lemma \ref{lem:manycomponents}. In particular, $\PP(E) \geq c_{p,k}''$, and under $E$, there are $c_{p,k}''\sqrt{t}^{0.35}$ different leaves, each one belonging to a different component. 
		Let $v$ such a leaf, and denote by $\tree^v_t$ the sub-tree rooted at $v$ at time $t$. Denote by $\tau$ the (random) number of insertions to $\tree^v_t$. By invoking Lemma \ref{lem:singlecomponentsize} on the sub-tree $\tree_t^v$, up to time $\tau^{\frac{0.35}{k}}$, we have 
		$$\PP\left(\exists t' \leq t\ \exists C\subset \tree^v_{t'} : |C| \geq  \tau^{0.35/k}\right) \geq \frac{c_{p,k}'}{\tau^{0.35}},$$ and by Lemma \ref{lem:leafevolution},
		$$\PP\left(\tau \geq \frac{1}{2}\sqrt{t} 
		\right)\geq \frac{1}{8}.$$
		Thus,
		$$\PP\left(\exists t' \leq t\ \exists C\subset \tree^v_{t'} : |C| \geq \frac{1}{2}\sqrt{t}^{0.35/k}\right) \geq \frac{c_{p,k}'}{8\sqrt{t}^{0.35}}.$$
		Since this is true for every different leaf, and since each sub-tree rooted at a leaf evolves independently, we have under $E$ that
		\begin{align*}
		\PP\left(\exists t' \leq t\ \exists C \in \CPT(t') : |C| \geq \frac{1}{2}\sqrt{t}^{\frac{0.35}{k}} \right) &\geq  1- \left(1-\frac{c_{p,k}'}{12\sqrt{t}^{0.35}} \right)^{|\CPT(\sqrt{t})|} \\
		&\geq  1 - \left(1-\frac{c_{p,k}'}{12\sqrt{t}^{0.35}} \right)^{c_{p,k}''\sqrt{t}^{0.35}}
		\\
		&\geq c_{p,k},
		\end{align*}
		where $c_{p,k} > 0$, depends only on $p$ and $k$.
	\end{proof}
	\section{Proofs for the General Model}\label{sec:generalproofs}
	We now turn to analyze the general model where $\eps > 0$ and new $\CF$ nodes may join the tree over time. The main differences caused by newly added $\CF$ nodes can be summarized by the following two points:
	\begin{itemize}
		\item When a check is performed, a removed node can potentially lie anywhere in the tree, while in the simple model, removed nodes are always roots of $\PT$ components. 
		\item If a new leaf is added to the tree at distance $k$ from a root of a $\PT$ component, it could be the case that the root will not be removed, even if a check is performed. This happens when the root has a $\CF$ descendent. 
	\end{itemize}
	Intuitively, the first point says that more nodes are removed when $\eps > 0$, which may reduce the odds of survival.
	On the other hand, the second point may give the impression that in the general model error effects can be harder to eliminate since roots have increased odds of survival.  
	Below, we address these conflicting views.
	
	\subsection{Error Effect Elimination in the General Model}
	We first investigate regimes where the error effects are completely eliminated.
	The following theorem, once established, will imply Theorem \ref{thm:error_effect_elim_general}.
	\begin{theorem}[Error effect elimination in the general model] \label{thm:generalerrorelim}
		Let $\eps,p \in (0,1)$ and $k \geq 2$ be such that 
		\begin{equation} \label{eq:superrelationelim}
		(1-\eps)\max\left(-\frac{1}{2}(2k-3)p + 3, -\frac{p}{2} + 3(1-p)\right) +2\eps(1-p) < 0.
		\end{equation}
		Then, the error effects in the $(\eps,p,k)$-CKP are completely eliminated. 
	\end{theorem}
	To see how Theorem \ref{thm:generalerrorelim} implies Theorem \ref{thm:error_effect_elim_general}, note that we can always choose $k$ large enough so that the relation in \ref{eq:superrelationelim} becomes,
	$$ (1-\eps)\left(-\frac{p}{2} + 3(1-p)\right) +2\eps(1-p) < 0.$$
	Now, since $\eps < 1$, the above inequality is satisfied for any $1 \geq p > \frac{4\eps-6}{5\eps-7},$ which yields Theorem \ref{thm:error_effect_elim_general}.
	
	The proof of Theorem \ref{thm:generalerrorelim} goes by extending the exponential potential to the general model. Before giving the definition, we introduce some new concepts.
	
	We shall say that a \False\ node $u$ is minimal, at time $t$, if $u$ is $\PT$, and it is either $\CF$, or its parent is $\PF$. In other words, $u$ is an active node, which is the root of a \False\ sub-component.
	Now, suppose that, at time $t$, $u$ is a minimal \False\ node in the $(\eps,p,k)$-CKP process. Denote by $\tree^u_t$ the tree rooted at $u$, and by $\CF(u,t)$ the set of all $\CF$ descendants of $u$,
	\begin{equation} \label{eq:CFdesc}
	\CF(u,t) =\{v\in \tree^u_t| \text{ there exists } u \neq w\in \tree^u_t \text{ such that } w \text{ is } \CF \text{ and } w \text{ is an ancestor of } v\}.
	\end{equation}
	We shall be interested in the sub-graph $\widetilde{\tree}^u_t:= \tree_t^u \setminus \CF(u,t)$. Before proceeding, to gain some intuition, observe that if $\{u_i\}_{i=1}^m$ is the collection of all minimal $\False$ nodes at some given time $t \geq 0$, then the sub-trees $\{\widetilde{\tree}_t^{u_i}\}_{i=1}^m$ partition all $\False$, with a $\PT$ label, into mutually disjoint trees, each one containing at most one $\CF$ node.
	
	Based on this partition we generalize the exponential potential to the general model.
	\begin{definition}[Exponential potential in the general model]
		Consider the $(\eps,p,k)$-CKP model at time $t \geq 0$, represented by $\tree_t$ and let $\mathcal{MF}$ be the set of all minimal \False\ nodes, at time $t$. For $u \in \mathcal{MF}$ and $v\in \widetilde{\tree}^u_t$,  the \emph{depth} $|v|$ is defined as the length (number of edges) of the shortest path between $v$ and $u$. The \emph{degree} is given by $d(v) := 1 + \deg_\PT(v)$ (note that $\deg_\PT(v)$ may depend on nodes outside of $\widetilde{\tree}^u_t$). 
		
		Define the \emph{exponential potential} of $\widetilde{\tree}^u_t$ by 
		$$
		\ExponentialPotential(\widetilde{\tree}^u_t) := \sum_{v \in \widetilde{\tree}^u_t} d(v) \cdot 2^{|v|},
		$$
		and 
		the potential of the CKP process, at time $t$, is defined as the sum over minimal $\False$ nodes
		$$\ExponentialPotential(\tree_t) := \sum\limits_{u \in \mathcal{MF}} \ExponentialPotential(\widetilde{\tree}^u_t).$$
	\end{definition}
	Observe that when $\eps = 0$, and $u$ is a minimal \False\ node at time $t$, the definition of $\widetilde{\tree}^u_t$ coincides with the definition of a $\PT$ components. Hence, the above definition is indeed a generalization of the exponential potential. With this extension, we may now prove Theorem \ref{thm:generalerrorelim}
	\begin{proof} [Proof of Theorem \ref{thm:generalerrorelim}]
		Consider the process $\{\ExponentialPotential(\tree_t)\}_{t\geq 0}$. Exactly like in the proof of Theorem \ref{thm:error_effect_elim_simple}, it will be enough to show that when $\ExponentialPotential(\tree_t) >0$, the process is a super-martingale. The final result will follow by applying the martingale convergence theorem.
		
		We thus focus on establishing that $\ExponentialPotential(\tree_t)$ is a super-martingale. Let $v$ be the new node added at time $t$. Denote $p(v)$ to be its parent, and assume for now that $v$ is $\CT$. Since we care about error elimination, it is fine to assume that the original root is $\CF$, and hence all nodes, including $p(v)$, are \False. Therefore, there exists a minimal \False\ node $u$, such that $v$ is added to $\widetilde{\tree}_t^u$. Observe that when a check is performed, only $u$ (and the path to $v$) can be removed. Thus, the analysis of Lemma \ref{lem:exponential_potential_contracts} applies and we deduce, for $D = \sum_{w \in \widetilde{T}^u_t} d(w)$, that
		$$\E\left[\left(\ExponentialPotential(\tree_{t+1}) - \ExponentialPotential(\tree_{t})\right){\bf1}_{\{v \text { is } \CT\}}|\tree_t\right]\leq (1-\eps) \frac{\ExponentialPotential(\widetilde{T}^u_t)}{D} \cdot \max\left(\left(-\frac{1}{2}(2k-3)p + 3\right), -\frac{p}{2} + 3(1-p)\right).$$
		The first term comes from \eqref{eq:exppotsuperfirst} and by noting $D \leq 
		\ExponentialPotential(\widetilde{T}^u_t)$, and the second term follows from \eqref{eq:exppotsmallcomps}. (Observe that the definitions of $d(v)$ and $D$ are slightly different from the simple CKP definitions used in Lemma \ref{lem:exponential_potential_contracts}; still, the same arguments as in the proof of that lemma follow, word for word, in the current setting.)

		Suppose now that $v$ is $\CF$ and that no check is performed (otherwise, $v$ would just remove itself when added). Here, $v$ creates a new component  $\widetilde{\tree}^v_t$, with $\ExponentialPotential(\widetilde{\tree}^v_t) = 1$. The insertion of $v$ also increases the potential of the component containing $\widetilde{T}^u_t$ by $2^{|p(v)|}$, since $v$ is $\PT$, meaning that its parent's degree increases by one. The expected increase in potential of $\widetilde{T}^u_t$ conditioning on this case is thus $$ \sum_{w\in\widetilde{T}^u_t} q(w) 2^{|w|} = 
		\frac{1}{D}\sum_{w\in\widetilde{T}^u_t} d(w) 2^{|w|} = 
		\frac{\ExponentialPotential(\widetilde{T}^u_t)}{D}.$$  
		
		Summarizing both of these contributions, we have
		$$
		\E\left[\left(\ExponentialPotential(\tree_{t+1}) - \ExponentialPotential(\tree_{t})\right){\bf1}_{\{v \text { is } \CF \text{ with no check}\}}|\tree_t\right] = \eps(1-p)\left(1+\frac{\ExponentialPotential(\widetilde{T}^u_t)}{D}\right) \leq 2\eps(1-p)\cdot \frac{\ExponentialPotential(\widetilde{T}^u_t)}{D}.$$
		
		We conclude that when 
		\begin{equation*} 
		(1-\eps)\max\left(-\frac{1}{2}(2k-3)p + 3, -\frac{p}{2} + 3(1-p)\right) +2\eps(1-p) < 0,
		\end{equation*}
		$\{\ExponentialPotential(\tree_t)\}_{t\geq 0}$ is a super-martingale, which concludes the proof.
	\end{proof}
	With Theorem \ref{thm:generalerrorelim}, we now identify a regime of parameters in which \True\ nodes are much more likely than \False\ nodes.
	\begin{theorem}
		Let $\eps, p \in (0,1)$ and $k \geq 2$ satisfy \eqref{eq:superrelationelim}.
		Further, Let $\tree_t$ be an $(\eps,p,k)$-CKP process and let $T_t$ and $F_t$ respectively stand for the subset of \True\ and \False\ nodes, labeled as $\PT$, in $\tree_t$, at time $t$. Then,
		$$\eps(1-p)\E\left[|T_t|\right] \geq (1-\eps)\E\left[|F_t|\right],$$
		for every $t\geq 0.$ In other words, the process is $O(\eps(1-p))$-highly reliable.
	\end{theorem}
	\begin{proof}
		Fix $t \geq 0$, and let $\mathcal{MF}(t)$ be the set of all minimal \False\ nodes, at time $t$. Consider the potential
		$$\Phi(\tree_t) = \frac{\eps(1-p)}{1-\eps}|T_t| - \sum\limits_{\tilde{\tree}^u_t: u\in \mathcal{MF}(t)}\ExponentialPotential(\tilde{\tree}^u_t).$$
		
		Now, let $v$ be the newly added node at time $t$, and let $p(v)$ be its parent. If $p(v)$ is a \False\ node, and if $\eps,p$, and $k$ satisfy \eqref{eq:superrelationelim}, then by the calculations in the proof of Theorem \ref{thm:generalerrorelim} we know that
		$$\E\left[\Phi(\tree_{t+1})- \Phi(\tree_t)|\tree_t\right] = \E\left[\sum\limits_{\tilde{\tree}^u_t: u\in \mathcal{MF}(t)}\ExponentialPotential(\tilde{\tree}^u_t) - \sum\limits_{\tilde{\tree}^u_{t+1}: u\in \mathcal{MF}(t+1)}\ExponentialPotential(\tilde{\tree}^u_{t+1})\right]>0.$$
		Otherwise, $p(v)$ is a \True\ node. In this case, $v$ is \CT\, with probability $1-\eps$, and then $|T_{t+1}| = |T_{t}|+1$, and $\Phi(\tree_{t+1}) -\Phi(\tree_t)=\frac{\eps}{1-\eps} $. The other possibility is that, with probability $\eps(1-p)$, $v$ is \CF\ and does no check, so $\mathcal{MF}(t+1) = \mathcal{MF}(t) \cup \{v\}$, and
		$\Phi(\tree_{t+1}) -\Phi(\tree_t) = -\ExponentialPotential(\{v\}) = -1.$ Thus, when $p(v)$ is \True,
		$$\E\left[\Phi(\tree_{t+1})- \Phi(\tree_t)|\tree_t\right] = (1-\eps)\frac{\eps(1-p)}{1-\eps} - \eps(1-p) =0.$$
		Altogether, we have shown that $\Phi(\tree_t)$ is a sub-martingale, so
		$\E\left[\Phi(\tree_t)\right] \geq \E\left[\Phi(\tree_0)\right] > 0.$
		Finally,
		$$\frac{\eps(1-p)}{1-\eps}|T_t| - |F_t| \geq \Phi(\tree_t) \implies \frac{\eps(1-p)}{1-\eps}\E[|T_t|] - \E[|F_t|] \geq 0 \implies \eps(1-p)\E[|T_t|]\geq (1-\eps)\E[|F_t|],$$
		which finishes the proof.
	\end{proof}
	
	\subsection{Survival in the General Model}
	As in the previous section, we begin by stating a more detailed version of Theorem \ref{thm:generalerrorsurvival}. 
	\begin{theorem}[Error effect survival in the general model] \label{thm:generalsurvival}
		For every $p,\eps \in [0,1]$ which satisfy
		$$\frac{1-p}{2}- 3(1-\eps)p > 0,$$
		and every $1 < k \leq \infty$, the error effects in the $(\eps,p,k)$-CKP survives.
	\end{theorem}
	Let us introduce some definitions to prepare for the proof of Theorem \ref{thm:generalsurvival}.
	Suppose that $u$ is a \CF\ node, in the $(\eps,p,k)$-CKP process, such that no ancestor of $u$ is \CF. Such a node will exist with probability $1$, for example, if it is the first $\CF$ node in the tree. As in the proof of Theorem \ref{thm:generalsurvival} denote by  $\tree^u_t$ the tree rooted at $u$ and by $\CF(u,t)$ the set of all $\CF$ descendants of $u$, as in \eqref{eq:CFdesc}.
	As before, we define $\widetilde{\tree}^u_t:= \tree_t^u \setminus \CF(u,t)$, but we now care about the dynamics of the forest $\widetilde{\tree}^u_t$ when $u$ stays fixed.
	Observe that $\widetilde{\tree}^u_t$ evolves like a (lazy) random tree with a different law than the preferential attachment tree. The discrepancy is caused by nodes in $\widetilde{\tree}_t^u$ which have descendants in $\CF(u,t)$, and hence have an increased probability, when compared to a preferential attachment tree, to spawn new nodes.
	We denote the $\PT$\ components of $\widetilde{\tree}_t^u$ by $\widetilde{\CPT}(t)$ and for $v \in \widetilde{\tree}_t^u$, we define $\deg_\CF(v)$ as the number of children of $v$ which lie in $\CF(u,t)$.
	
	To prove Theorem \ref{thm:generalsurvival}, we make the following generalization of the leaves and components potential, to the general model.
	
	\begin{definition}[Leaves and components potential in the general model]\label{def:adapted_comps_leaves_potential}
		A node $v$ in a component $C \in \widetilde{\CPT}(t)$ is considered a leaf if it does not have any descendent in $C$ (we note that a leaf is allowed to have descendants in $\CF(u,t)$). 
		
		For a given component $C \in \widetilde\CPT(t)$, the \emph{leaves and components potential} restricted to $C$ is $1$ if $|C|=1$, and otherwise it is, 
		$$1 + \sum\limits_{v \in C \text{ is a leaf}}\frac{1}{\deg_\CF(v) + 1}.$$ 
		For the sub-tree $\widetilde{\tree}_t^u$, the leaves and components potential $\AdapatedLeavesComponentsPotential(\widetilde{\tree}_t^u)$ is the sum of potentials of all $C \in \widetilde\CPT(t)$.
		Finally, we define the leaves and components potential for the entire tree $\tree_t^u$ recursively in the following way. Let $\{u_i\}_{i=1}^m$ the set of all nodes in $\CF(u,t)$ with parent in $\widetilde{\tree}_t^u$, then,
		$$\AdapatedLeavesComponentsPotential(\tree_t^u)=\AdapatedLeavesComponentsPotential(\widetilde{\tree}_t^u) + \sum\limits_{i=1}^m\AdapatedLeavesComponentsPotential(\widetilde{\tree}_t^{u_i}).$$
	\end{definition}
	We remark that if $v$ is a leaf according to the above definition in $\widetilde{\tree}^u_t$, then it has no $\CF$ children (by definition of $\widetilde{\tree}^u_t$) and no $\PT$ children (as a leaf), that is, $\deg_\CF(v) = \deg_\PT(v) = 0$. Thus, if $\eps = 0$, the above definition agrees with the definition of $\LeavesComponentsPotential$ from the simple model.
	
	We now prove a generalization of Lemma \ref{lem:leavespotentialevol} to the general model.
	\begin{lemma} \label{lem:adaptedleavespotentialevol}
		Consider the process $(\widetilde{\tree}^u_t)_{t \geq 0}$, generated from the $(\eps, p,k)$-process, let $t \in \N$, and suppose that $\AdapatedLeavesComponentsPotential(\widetilde{\tree}_t^u) > 0$. Denote by $\Delta_t = \AdapatedLeavesComponentsPotential(\widetilde{\tree}_{t+1}^u) - \AdapatedLeavesComponentsPotential(\widetilde{\tree}_t^u)$ the change in the leaves and components potential at time $t$. Then $\Delta_t \geq  -2$ always holds, and further 
		$$\E[\Delta_t | \widetilde{\tree}_t^u] \geq \frac{1-p}{2}- 3(1-\eps)p.$$ 
	\end{lemma}
	
	\begin{proof}
		Let $v$ denote the node added to the process at time $t+1$. Denote its parent by $p(v)$ and let $C \in \widetilde\CPT(t)$ be the $\PT$ component containing $p(v)$ at time $t$. 
		Note that attaching $v$ to $C$ does not modify any of the PT components $C' \neq C$. 
		Therefore suffices to analyze the change of potential in $C$.
		
		First, if $|C| = 1$, then $|C \cup \{v\}| = 2$. Suppose first that $v$ does not run the checking procedure; this holds with probability $1-p$. Since $C \cup \{v\}$ contains precisely one leaf in this case, its total potential is $2$, an increase of $1$ over the potential of $C$.
		In the other case, with probability at most $p$, checking takes place and both $u$ and $v$ are marked PF, thus removing $C$ from $\widetilde\CPT(t+1)$ without creating new \PT\ nodes, which decreases the total potential by $1$. In total, the expected change in potential is at least $1 \cdot (1-p) - 1 \cdot p = 1-2p$.
		
		Otherwise, $|C| > 1$, and $\deg_{\PT}(\mathrm{root}) \geq 1$. Since $v$ chooses its parent  $p(v)$ according to the \emph{preferential attachment distribution over $\tree^u_t$}, we have,
		$$\alpha:=\PP\left(p(v) \text{ is a leaf}\right) = \frac{\sum\limits_{w\in C \text{ is a leaf}}(\deg_{\PF}(w)+1)}{\sum\limits_{w\in C}(\deg_{\PT}(w) +1)} = \frac{\sum\limits_{w\in C \text{ is a leaf}}(\deg_{\PF}(w)+1)}{q(t)},$$
		where we have set $q(t):= \sum\limits_{w\in C}(\deg_{\PT}(w) +1)$.
		There are several cases to consider. We begin with the possibilities to increase the potential. For this, it will be helpful to note that if the new node, $v$, is $\CT$, then $\deg_\PF(u) = 0$, and so it contributes $1$ to the potential. Thus let us denote the event $E:=\{v \text{ is } \CT \text{ and no check was performed}$. It is immediate that $\PP(E) =(1-\eps)(1-p). $
		\begin{enumerate}
			\item If $p(v)$ is a leaf and $v$ is $\CT$, then after addition $p(v)$ will no longer be a leaf. Thus, with probability $1-p$ no check was performed, and so, when $E$ happens, the expected increase in the potential is at least,
			\begin{align*}
			(1-\eps)&(1-p)\sum\limits_{w\in C \text{ is a leaf}} \PP\left(p(v) = w\right)\left(1- \frac{1}{\deg_{\CF}(w) + 1}\right)\\
			&= \frac{  (1-\eps)(1-p)}{q(t)}\sum\limits_{w\in C \text{ is a leaf}} \left(\deg_{\CF}(w) + 1\right)\left(1- \frac{1}{\deg_{\CF}(w) + 1}\right)\\
			&= (1-\eps)(1-p)\left(\alpha - \frac{\#\{\text{leaves in } C\}}{q(t)}\right)
			.
			\end{align*}
			\item If $p(v)$ is not a leaf, under $E$, the expected increase in the potential is at least,
			\begin{align*}
			(1-\eps)(1-p)\PP\left(p(v) \text{ is not a leaf}\right) =  (1-\eps)(1-p)(1-\alpha)
			\end{align*}
		\end{enumerate}
		Combining the above two cases we see,
		$$\E\left[\Delta_t{\bf 1}_E|\AdapatedLeavesComponentsPotential(t)\right] \geq (1-\eps)(1-p)\left(1 - \frac{\#\{\text{leaves in } C\}}{q(t)}\right) \geq \frac{(1-\eps)(1-p)}{2}.$$
		Now let us address the cases where the potential may decrease.
		\begin{enumerate}
			\item Suppose that $v$ is $\CT$, but that it runs a check (with probability $p$), then the potential can decrease in two ways.
			The added node $v$ can remove a leaf from $C$, which can only happen if $p(v)$ is a leaf.
			Note that any removed parent of $p(v)$, other than the root, can only increase the potential, since it would create new connected components, without affecting the number of leaves.
			So, the other possibility to decrease the potential is to remove the root. This can happen regardless of whether $v$ is connected to a leaf or an internal node. 
			Thus,
			$$\PP\left(2\leq \Delta_t < -1 \text{ and } v \text{ is } \CT \right)\leq (1-\eps)\PP\left(u \text{ is a leaf and a check was performed} \right) \leq (1-\eps)p,$$
			and 
			$$\PP\left(\Delta_t = -1 \text{ and } v \text{ is } \CT \right)\leq (1-\eps)\PP\left(u \text{ is not a leaf and a check was performed} \right) \leq (1-\eps)p,$$
			\item The final possibility is that $v$ is $\CF$, which does not run a check, and that $p(v)$ is a leaf. In this case the weight $\frac{1}{\deg_\PF(p(v))+1}$ is going decrease to $\frac{1}{\deg_\PF(p(v))+2}$, and
			$$\frac{1}{\deg_\PF(p(v))+2} - \frac{1}{\deg_\PF(p(v))+1} \geq \frac{1}{2}.$$
			On the other hand now, $v \in \CF(u,t+1)$ and is the root of a new tree $\widetilde{\tree}^v_{t+1}$ with $\AdapatedLeavesComponentsPotential(\widetilde{\tree}^v_{t+1}) = 1$.
			Thus, the expected change is at least 
			$$\PP\left(v \text{ is } \CF \text{ and no check was performed}\right)\left(1-\frac{1}{2}\right) =\frac{\eps(1-p)}{2}.$$
		\end{enumerate}
		The above calculations show,
		$$\E\left[\Delta_t{\bf 1}_{E^c}|\widetilde{\tree}_t^u\right] \geq \frac{\eps(1-p)}{2} -2p(1-\eps) -p(1-\eps) = \frac{\eps(1-p)}{2} -3(1-\eps)p.$$
		Altogether we see,
		$$\E\left[\Delta_t|\widetilde{\tree}_t^u\right] \geq \frac{(1-\eps)(1-p)}{2}+\frac{\eps(1-p)}{2} -3(1-\eps)p = \frac{1-p}{2}- 3(1-\eps)p.$$
	\end{proof}
	
	\begin{lemma} \label{lem:singletreesurvival}
		Suppose that $\frac{1-p}{2}- 3(1-\eps)p > 0$, then $\{\tree^u_t\}_{t\geq 0}$ survives with positive probability.
	\end{lemma}
	\begin{proof}
		Write $X_t := \AdapatedLeavesComponentsPotential(\tree^u_t)$. Since $|\tree^u_t| \geq X_t$, it will be enough to show,
		$$\PP\left(\min\limits_{t\geq 0} X_t> 0\right) > 0.$$
		Keeping this in mind, with no loss of generality, it is fine to assume $X_0 = C_{\eps,p}$, for some large constant $C_{\eps,p} > 0$. This is because every finite configuration happens with positive probability.
        Similar to the case of Theorem \ref{thm:simpleerrorsurvival}, we need a sub-martingale with bounded increments. In a similar vein we define $\tilde{X}_t$ to satisfy $\tilde{X}_0 = X_0$ and 
         $$\tilde{X}_t = \begin{cases}\tilde{X}_{t-1} + X_t - X_{t-1}& \text{if } X_t - X_{t-1} \leq 2\\ 
        \tilde{X}_{t-1} + 2 & \text{if } X_t - X_{t-1} > 2\end{cases}.$$
		
		Repeating the same arguments as in the proof of Theorem \ref{thm:simpleerrorsurvival}, Lemma \ref{lem:adaptedleavespotentialevol} implies that, under the condition $c:= \frac{1-p}{2}- 3(1-\eps)p > 0$, $\tilde{X}_t$ is a sub-martingale such that,
		$|\tilde{X}_{t+1} - \tilde{X}_t| \leq 2,$ and when $X_t \neq 0$, $\E\left[\tilde{X}_{t+1} -\tilde{X}_t|X_t\right] \geq c>0$. The claim now follows by invoking Lemma \ref{lem:positivesubmartingale}, and noting $X_t \geq \tilde{X}_t$, almost surely.
	\end{proof}
	Given Lemma \ref{lem:singletreesurvival}, Theorem \ref{thm:generalsurvival} readily follows.
	\begin{proof}[Proof of Theorem \ref{thm:generalsurvival}]
		Let $u$ be any $\CF$ node in the $(\eps,p,k)$-CKP process. Then if $$\frac{1-p}{2}- 3(1-\eps)p > 0,$$
		Lemma \ref{lem:singletreesurvival} shows that $\tree^u_t$ survives with positive probability and we may conclude that the error effect survives in the $(\eps,p,k)$-CKP process.
	\end{proof}

	\section{Discussion and Open Questions} \label{sec:open_questions}
	In this paper, we have studied the different behaviors of CKPs when the underlying graph process is a tree. Our results reveal striking differences in the overall shape and persistence of the processes. These differences depend on the interactions between the different parameters.
	
	When $k$ is large, Theorems \ref{thm:error_effect_elim_simple} and \ref{thm:simpleerrorsurvival}, for the simple model, along with Theorems \ref{thm:error_effect_elim_general} and \ref{thm:generalerrorsurvival}, for the general model, identified a phase transition for the error effects, which mainly depends on $p$. The results establish that the propagation of errors can be completely eliminated, with absolute certainty, as long as we put some constant, which is necessarily not too small, fraction of knowledge units under scrutiny. In contrast, in Theorem \ref{thm:notwolevels}, we elucidated the dramatic role of the depth, $k$, and showed that very shallow checks completely nullify the above dependence in $p$. In particular, when $k=2$, there is no way to guarantee the elimination of error effects, which should discourage shallow checking procedures. 
	
	Other than considering phase transitions, we have also studied the structural properties of the processes in the different regimes. In Theorem \ref{thm:reliableprocess}, we focused on the case of small $p$, where the error effects in the simple process can survive. According to the theorem, as long as $p$ is not small, as dictated by $k$, even when the simple model survives, one may still guarantee that no single error can be connected to most of the entire process. In particular, each surviving component will only have sub-linear size. Finally, Theorem \ref{thm:highreliableprocess}, dealt with a regime of the general mode in which error effects are guaranteed to be eliminated. By design, even when error effects are eliminated, the general model continues to evolve, and we show that, by making $p$ still larger, we may also guarantee that the proportion of $\False$ remains almost minimal.

	While we aimed to cover the wide range of possible phenomena exhibited by the CKPs, our work also leaves some open questions. Below we list several such questions and other possible directions for research.

	\begin{itemize}
		\item \textbf{Shallowness of checks:} As detailed above, for the simple model, there is a strict phase transition, depending on $k$. However, our results do not cover the cases $k = 3$ or $k=4$, and it will be interesting to see whether the error effects can be eliminated in this case. 
		\begin{question}
			Is there some $p < 1$, such that the error effect in the $(p,3)$-simple CKP is completely eliminated? What about $(p,4)$-simple CKP?
		\end{question}
		It seems that a positive answer to the above question would require a more subtle potential than our exponential potential.
		
		\item \textbf{Critical parameters:} 
		In a similar vein to the previous question, the, arguably challenging, question of finding the critical $p$ for the transition remains open. 
		\begin{question}
			What is the value of $p_0 \in (0,1)$ (which may depend on $k$), such that for any $p<p_0$ the error effect in the $(p,k)$-simple CKP survives with positive probability, and for $p>p_0$ the error effect in the $(p,k)$-simple CKP is completely eliminated?
		\end{question}
		Similar questions are left open with respect to our other definitions. For example, finding the correct polynomial power in Theorem \ref{thm:reliableprocess} is also of interest.
		
		\item \textbf{Proportion of false nodes:} 
		Another possible direction would be to complete the picture presented in Theorem \ref{thm:highreliableprocess} and show a result in the converse direction.
		
		\begin{question}
			Is there a set of non-trivial parameters $p,k,\eps$, such that the expected proportion of \False\ nodes in the $(\eps,p,k)$-CKP is $1 - o(1)$?
		\end{question}
		
		Note that for $\True$ nodes, unlike their $\False$ counterparts, there is no a priori probabilistic guarantee on their expected proportion. Thus, it makes sense to study regimes where all nodes, except a negligible proportion, are $\False$. In particular, identifying the possible existence of intermediate regimes where $\False$ nodes exist in abundance, yet do not overwhelm the process, is also of interest.
		\item \textbf{More realistic models:}  As discussed in the introduction, we considered a simplified model of knowledge accumulation, where each knowledge unit relies on a single existing previous unit. This restrictive assumption naturally leads to the preferential attachment tree we've considered. However, in many cases of interest, trees do not necessarily provide a faithful representation of knowledge accumulation since new knowledge can rely on several different sources, which leads to directed acyclic graphs.
		\begin{question}
			Can similar results apply in more general CKPs where the underlying model is a DAG and nodes can have an in-degree larger than $1$?
		\end{question}
		The crucial point is that any extension of our model to general DAGs must also specify a natural way for a new node to choose a set of `parents'. Unlike the tree model, it is not satisfactory to choose a random subset of existing vertices since new units should be more likely to rely on existing units that are similar, in some sense to be defined. So, a general model should define a similarity metric on nodes and choose new parents based on both their degrees and this metric, leading to a more involved analysis. Towards this aim, a subsequent work, \cite{brandenberger2023combinative}, including a subset of the authors, showed that many of the phenomena shown in this paper also extend to more general settings. In particular, the new models allow for DAG-like CKPs and more general classes of growth models.
	\end{itemize}

	\bibliographystyle{alpha}
	\bibliography{references}

\begin{thebibliography}{DvdHH10}

\bibitem[AMP20]{AMP:20}
Noga {Alon}, Elchanan {Mossel}, and Robin {Pemantle}.
\newblock {Distributed Corruption Detection in Networks}.
\newblock {\em Theory of Computing}, 16(1):1--23, 2020.

\bibitem[ASS18]{allison2018reproducibility}
David~B Allison, Richard~M Shiffrin, and Victoria Stodden.
\newblock Reproducibility of research: Issues and proposed remedies.
\newblock {\em Proceedings of the National Academy of Sciences},
  115(11):2561--2562, 2018.

\bibitem[BL12]{brightwell2012vertices}
Graham Brightwell and Malwina Luczak.
\newblock Vertices of high degree in the preferential attachment tree.
\newblock {\em Electronic Journal of Probability}, 17:no. 14, 43, 2012.

\bibitem[BMMS24]{brandenberger2023combinative}
Anna Brandenberger, Cassandra Marcussen, Elchanan Mossel, and Madhu Sudan.
\newblock Errors are robustly tamed in cumulative knowledge processes.
\newblock In {\em Proceedings of Thirty Seventh Conference on Learning Theory},
  volume 247 of {\em Proceedings of Machine Learning Research}, pages 630--631.
  PMLR, 2024.

\bibitem[BMS21]{information-spread}
Omri Ben{-}Eliezer, Elchanan Mossel, and Madhu Sudan.
\newblock Information spread with error correction.
\newblock {\em arXiv preprint}, page 2107.06362, 2021.

\bibitem[Doe11]{doerr}
Benjamin Doerr.
\newblock Analyzing randomized search heuristics: {Tools} from probability
  theory.
\newblock In {\em Series on Theoretical Computer Science}, pages 1--20. World
  Scientific, February 2011.

\bibitem[Dur19]{durrett2016probability}
Rick Durrett.
\newblock {\em Probability---theory and examples}, volume~49 of {\em Cambridge
  Series in Statistical and Probabilistic Mathematics}.
\newblock Cambridge University Press, Cambridge, 2019.
\newblock Fifth edition of [ MR1068527].

\bibitem[DvdHH10]{dommers2010diameters}
Sander Dommers, Remco van~der Hofstad, and Gerard Hooghiemstra.
\newblock Diameters in preferential attachment models.
\newblock {\em Journal of Statistical Physics}, 139(1):72--107, 2010.

\bibitem[ES99]{EvansSchulman:99}
William~S. Evans and Leonard~J. Schulman.
\newblock Signal propagation and noisy circuits.
\newblock {\em IEEE Trans. Inform. Theory}, 45(7):2367--2373, 1999.

\bibitem[Fol99]{folland1999real}
Gerald~B. Folland.
\newblock {\em Real analysis}.
\newblock Pure and Applied Mathematics (New York). John Wiley \& Sons, Inc.,
  New York, second edition, 1999.
\newblock Modern techniques and their applications, A Wiley-Interscience
  Publication.

\bibitem[Gra01]{gray2001reader}
Lawrence~F. Gray.
\newblock A reader's guide to {G}acs's “positive rates” paper.
\newblock {\em Journal of Statistical Physics}, 103(1):1--44, 2001.

\bibitem[Grc13]{Grcar2013}
Joseph~F. Grcar.
\newblock Errors and corrections in mathematics literature.
\newblock {\em Notices of the AMS}, 60(4):418--425, 2013.

\bibitem[Ioa05]{Ioannidis2005}
John P.~A. Ioannidis.
\newblock Why most published research findings are false.
\newblock {\em PLOS Medicine}, 2(8):e124, 2005.

\bibitem[Ioa12]{IannidisSelfCorrect2012}
John P.~A. Ioannidis.
\newblock Why science is not necessarily self-correcting.
\newblock {\em Perspectives on Psychological Science}, 7(6):645--654, 2012.

\bibitem[Lon19]{sep-scientific-knowledge-social}
Helen Longino.
\newblock {The Social Dimensions of Scientific Knowledge}.
\newblock In Edward~N. Zalta, editor, {\em The {Stanford} Encyclopedia of
  Philosophy}. Metaphysics Research Lab, Stanford University, {S}ummer 2019
  edition, 2019.

\bibitem[LR18]{LarcombeEnvironmental2018}
Piers Larcombe and Peter Ridd.
\newblock The need for a formalised system of quality control for environmental
  policy-science.
\newblock {\em Marine Pollution Bulletin}, 126:449--461, 2018.

\bibitem[MMP20]{MaMoPo:20}
Anuran Makur, Elchanan Mossel, and Yury Polyanskiy.
\newblock Broadcasting on random directed acyclic graphs.
\newblock {\em IEEE Transactions on Information Theory}, 66(2):780--812, 2020.

\bibitem[Pil22]{blots22}
Charles Piller.
\newblock Blots on a field.
\newblock {\em Science}, 377(6604):358--363, 2022.

\bibitem[PMC67]{PMC:67}
Franco~P. Preparata, Gernot Metze, and Robert~T. Chien.
\newblock On the connection assignment problem of diagnosable systems.
\newblock {\em IEEE Transactions on Electronic Computers}, 6(EC-16):848--854,
  1967.

\bibitem[PS22]{pain2022correction}
Michel Pain and Delphin S\'{e}nizergues.
\newblock Correction terms for the height of weighted recursive trees.
\newblock {\em Ann. Appl. Probab.}, 32(4):3027--3059, 2022.

\bibitem[SC22]{SelkoeCummings22}
Dennis Selkoe and Jeffrey Cummings.
\newblock News story miscasts {A}lzheimer's science.
\newblock {\em Science}, 377(6609):934--935, 2022.

\bibitem[vN56]{vonNeumann:56}
John von Neumann.
\newblock Probabilistic logics and the synthesis of reliable organisms from
  unreliable components.
\newblock In {\em Automata studies}, Annals of mathematics studies, no. 34,
  pages 43--98. Princeton University Press, Princeton, N. J., 1956.

\bibitem[Yul25]{yule1925mathematical}
George~Udny Yule.
\newblock A mathematical theory of evolution, based on the conclusions of {Dr.
  JC Willis, FR S}.
\newblock {\em Philosophical transactions of the Royal Society of London.
  Series B, containing papers of a biological character}, 213(402-410):21--87,
  1925.

\end{thebibliography}
\end{document}